\documentclass[letterpaper]{article} 
\usepackage{aaai25}  
\usepackage{times}  
\usepackage{helvet}  
\usepackage{courier}  
\usepackage[hyphens]{url}  
\usepackage{graphicx} 
\usepackage{natbib}  
\usepackage{caption} 
\usepackage{algorithm}
\usepackage{algorithmic}

\usepackage{newfloat}
\usepackage{listings}
\DeclareCaptionStyle{ruled}{labelfont=normalfont,labelsep=colon,strut=off} 
\floatstyle{ruled}
\newfloat{listing}{tb}{lst}{}
\floatname{listing}{Listing}
\title{LTLf+ and PPLTL+: Extending LTLf and PPLTL to Infinite Traces}

\author {%
Benjamin Aminof\textsuperscript{\rm 1},
Giuseppe De Giacomo\textsuperscript{\rm 1,2},
Sasha Rubin\textsuperscript{\rm 3},
Moshe Y.\ Vardi\textsuperscript{\rm 4}
}
\affiliations {
\textsuperscript{\rm 1}University of Rome ``La Sapienza'', Italy\\
\textsuperscript{\rm 2}University of Oxford, United Kingdom\\
\textsuperscript{\rm 3}University of Sydney, Australia\\
\textsuperscript{\rm 4}Rice University, USA\\
benjamin.aminof@tuwien.ac.at, giuseppe.degiacomo@cs.ox.ac.uk, sasha.rubin@sydney.edu.au,
vardi@cs.rice.edu
}

\usepackage{xcolor}
\usepackage{comment}
\usepackage{amsthm}
\usepackage{amsmath}
\usepackage{amssymb}
\usepackage{xspace}

\newtheorem{theorem}{Theorem}
\newtheorem{lemma}{Lemma}
\newtheorem{example}{Example}
\newtheorem{remark}{Remark}
\newtheorem{corollary}[theorem]{Corollary}
\newtheorem{proposition}{Proposition}
\newtheorem{definition}{Definition}

\newcommand{\DFA}{{DFA}\xspace}

\def\AP{{AP}}

\def\tf{\{\true,\false\}}

\newcommand\exptime{EXPTIME\xspace}

\newcommand{\limp}{\supset}

\DeclareMathOperator{\LTLU}{\mathbin{\mathsf{U}}}
\DeclareMathOperator{\LTLW}{\mathbin{\mathsf{W}}}
\DeclareMathOperator{\LTLR}{\mathbin{\mathsf{R}}}
\DeclareMathOperator{\LTLM}{\mathbin{\mathsf{M}}}
\DeclareMathOperator{\LTLX}{\mathbin{\mathsf{X}}}
\DeclareMathOperator{\LTLWX}{\mathbin{\mathsf{WX}}}

\DeclareMathOperator{\LTLS}{\mathbin{\mathsf{S}}}
\DeclareMathOperator{\LTLY}{\mathbin{\mathsf{Y}}}
\DeclareMathOperator{\LTLF}{\mathbin{\mathsf{F}}}
\DeclareMathOperator{\LTLG}{\mathbin{\mathsf{G}}}
\DeclareMathOperator{\LTLH}{\mathbin{\mathsf{H}}}
\DeclareMathOperator{\LTLO}{\mathbin{\mathsf{O}}}

\DeclareMathOperator{\nextX}{\mathsf{X}}

\DeclareMathOperator{\eventually}{\ensuremath{\mathsf{F}}\xspace}

\newcommand{\true}{\mathsf{true}}
\newcommand{\false}{\mathsf{false}}

\newcommand{\BOX}[1]{ [#1]}
\newcommand{\DIAM}[1]{\langle #1 \rangle}

\newcommand{\AUTA}{\mathcal{A}}
\newcommand{\AUTB}{\mathcal{B}}
\newcommand{\AUTC}{\mathcal{C}}

\definecolor{darkblue}{rgb}{0.0, 0.0, 0.55}
\definecolor{darkorange}{rgb}{1.0, 0.55, 0.0}

\newcommand\len[1]{|{#1}|}
\newcommand\emptystring{\varepsilon}

\newcommand\trace{\tau}

\renewcommand\inf{\textit{inf}}

\newcommand{\LTL}{\textrm{{LTL}}\xspace}
\newcommand{\LTLf}{\ensuremath{{\textrm{LTLf}}}\xspace}
\newcommand{\PPLTL}{\ensuremath{{\textrm{PPLTL}}}\xspace}
\newcommand{\LTLfplus}{\ensuremath{{\textrm{LTLf+}}}\xspace}
\newcommand{\PPLTLplus}{\ensuremath{{\textrm{PPLTL+}}}\xspace}

\newcommand{\LDLf}{\ensuremath{{\textrm{{LDLf}}}}\xspace}
\newcommand{\PPLDL}{\ensuremath{{\textrm{{PPLDL}}}}\xspace}

\newcommand\last{\textit{last}}

\newcommand\safe[1]{\ensuremath{\forall #1}}
\newcommand\guar[1]{\ensuremath{\exists #1}}
\newcommand\pers[1]{\ensuremath{\exists\forall #1}}
\newcommand\recu[1]{\ensuremath{\forall\exists #1}}

\newcommand\safeSymbol{\ensuremath{\forall}}
\newcommand\guarSymbol{\ensuremath{\exists}}
\newcommand\persSymbol{\ensuremath{\exists\forall}}
\newcommand\recuSymbol{\ensuremath{\forall\exists}}

\newcommand{\QeL}{\mathfrak{L}+}

\usepackage{amsmath}
\usepackage{mathabx}

\usepackage{todonotes}

\newcommand{\srnotes}[1]{\todo[inline,color=blue!10,caption={SR}]{\textbf{SR:} #1}}

\begin{document}

\maketitle




\begin{abstract}
We introduce LTLf+ and PPLTL+, two logics to express properties of infinite traces, that are based on the linear-time temporal logics LTLf and PPLTL on finite traces. LTLf+/PPLTL+ use levels of Manna and Pnueli’s LTL safety-progress hierarchy, and thus have the same expressive power as LTL. However, they also retain a crucial characteristic of the reactive synthesis problem for the base logics: the game arena for strategy extraction can be derived from deterministic finite automata (DFA). Consequently, these logics circumvent the notorious difficulties associated with determinizing infinite trace automata, typical of LTL reactive synthesis. We present DFA-based synthesis techniques for LTLf+/PPLTL+, and show that synthesis is 2EXPTIME-complete for LTLf+ (matching LTLf) and EXPTIME-complete for PPLTL+ (matching PPLTL).  Notably, while PPLTL+ retains the full expressive power of LTL, reactive synthesis is EXPTIME-complete instead of 2EXPTIME-complete. The techniques are also adapted to optimally solve satisfiability, validity, and model-checking, to get EXPSPACE-complete for LTLf+ (extending a recent result for the guarantee level using LTLf), and PSPACE-complete for PPLTL+.
\end{abstract}




\section{Introduction}
        
\emph{Reactive synthesis} is concerned with synthesizing programs (aka, strategies) for reactive computations
(e.g., processes, protocols, controllers, robots) in active environments \cite{PnueliR89,finkbeiner2016synthesis,EhlersLTV17}.
The basic techniques for reactive synthesis share several similarities with Model Checking, and are based on the connections between Logics, Automata, and Games
\cite{fijalkow2023games}. The most common specification language is possibly Linear Temporal Logic (\LTL)~\cite{Pnueli77}.

Reactive Synthesis for \LTL involves the following Steps: (1) having a specification $\varphi$ of the desired system behavior in \LTL,
in which one distinguishes controllable and uncontrollable variables; 
(2) extracting from the specification an equivalent automaton on infinite words, corresponding to the infinite traces satisfying $\varphi$;
(3) (differently from Model Checking) determinizing the automaton to obtain an arena for a game between the system and the environment;
(4) solving the game, by fixpoint computation, for an objective determined by the automaton's accepting condition (e.g., a parity objective for \LTL), yielding a strategy for the system that fulfills the original specification $\varphi$.

Model Checking is mature, {and many of its techniques may be exploited in Reactive Synthesis as well}, e.g., symbolic techniques based on Boolean encodings may be used to compactly represent the game arena and to compute fixpoints over it. 
{However, despite this,} Step (3) remains a major performance obstacle. For \LTL, this involves determinizing nondeterministic B\"uchi automata, which is notoriously difficult \cite{Vardi07}. This has held back the use of reactive synthesis in applications.  

Reactive synthesis is deeply related to Planning \cite{DR-IJCAI18,CamachoBM19}, and in particular to (strong) planning for temporally extended goals in fully observable nondeterministic domains \cite{Cimatti03,DBLP:journals/amai/BacchusK98,BacchusK00,CalvaneseGV02,BaierM06,BaierFM07,GereviniHLSD09}.
{A key characteristic of Planning is that} the system continuously receives a goal, ``thinks'' about how to achieve it, synthesizes a plan, executes the plan, and  repeats~\cite{GeBo13}. 
This suggests to focus on goal specifications that can be satisfied on finite traces.
%
%
Recently, this led to a shift in Reactive Synthesis to focus on logics on finite traces (instead of infinite traces), e.g., \LTLf~\cite{DegVa13,DegVa15}. 
The advantage of focusing on finite traces is that in Step (3) one can rely on (classic) automata operating on finite traces, including deterministic finite automata (DFA), and use known determinization algorithms with good practical performance. 
%
The development of \LTLf synthesis~\cite{DegVa15} has brought about scalable tools that are unprecedented in reactive synthesis \cite{ZTLPV17,bansal2020hybrid,DF2021,DeGiacomoFLVX022}.
Beside \LTLf, another finite-trace logic that is gaining popularity in AI is {\em Pure Past} \LTL (\PPLTL)~\cite{DeGiacomoSFR20,CimattiGGMT20,BonassiGFFGS23icaps,BonassiGFFGS23ecai,BonassiDGS24ecai}. This is a variant of \LTLf that sees the trace backwards and has the notable property that one can obtain a symbolic (i.e., factorized) DFA directly from the formula in linear time; moreover, while the size of the (non-symbolic) DFA corresponding to an \LTLf formula can be double-exponential in the size of the formula itself, the size of the DFA corresponding to a \PPLTL formula is at most a single-exponential in the size of the formula. 

Nevertheless, not all specifications of interest can be expressed {on finite traces}. {For example, the planning domain is an infinite-trace specification: the planning domain will continue to respond to actions (with preconditions satisfied) by producing its possibly nondeterministic effects, 
\emph{forever}.  Not to mention \emph{recurrence},   \emph{persistence}, \emph{reactivity}, and other properties typically used in Model Checking. When dealing with infinite traces, using (same variants of ) \LTL as the specification language is an obvious choice.}

\emph{Can we lift the DFA techniques at the base of the success story of \LTLf and \PPLTL synthesis to full \LTL?
In this paper we answer this question positively!}

To do so, we leverage the classic hierarchy of \LTL properties --- the \emph{safety-progress} hierarchy \cite{DBLP:conf/podc/MannaP89}.\footnote{
The hierarchy was introduced by \citeauthor{LichtensteinPZ85} in \citeyear{LichtensteinPZ85}, later described in detail by \citeauthor{DBLP:conf/podc/MannaP89} in \citeyear{DBLP:conf/podc/MannaP89} and in their books \cite{MannaPnueli92,MannaPnueli95,MannaPnueli10}; also, see the survey~\cite{PitermanP18}.} It consists of six classes of semantic properties, organized by inclusion. The bottom, first, level has, on the one hand, the \emph{safety properties} (that intuitively express that nothing bad
ever happens), and on the other the \emph{guarantee properties}, sometimes also called co-safety properties, (that express that something good eventually happens); the second level consists of the class of \emph{obligation properties}, obtained as positive Boolean combination of safety and guarantee properties; the third level contains, on the one hand, the \emph{recurrence properties} (that express that something good occurs infinitely often), and on the other the \emph{persistence properties} (that say that nothing bad occurs infinitely often); and the fourth level contains the \emph{reactivity properties}, which are obtained as positive Boolean combination of recurrence and persistence properties. Each property is semantically defined in terms of sets of finite traces, e.g., a set $F$ of finite-traces induces a basic safety (resp. progress) property that consists of an infinite trace iff every prefix (resp. all but finitely many prefixes) of the trace are in~$F$.
The reactivity properties contain all properties expressible in \LTL.\footnote{The hierarchy is not limited to \LTL, i.e., to properties that are expressible in first-order logic (FO) over infinite sequences~\cite{Kamp}, but extends to omega-regular properties, i.e., to monadic-second order logic (MSO) over infinite sequences. Indeed all the results we present here can be extended to omega-regular properties by substituting \LTLf (resp. \PPLTL) by its MSO-complete variant \LDLf (resp. \PPLDL) \cite{DegVa13}.}

We revisit Manna and Pnueli's hierarchy, and exploit it to define extensions of \LTLf and \PPLTL, which we call \LTLfplus and \PPLTLplus, that can express arbitrary \LTL properties on infinite traces. 
These new logics retain a crucial characteristic for reactive synthesis of their base logics: one can exploit the techniques for translating \LTLf and \PPLTL formulas into DFAs \cite{DegVa15,DeGiacomoSFR20}.  These DFAs combine in a simple product to form the game arena for strategy extraction of \LTLfplus/\PPLTLplus specifications. 
Naturally, the game objectives for \LTLfplus/\PPLTLplus go beyond the simple adversarial reachability for \LTLf/\PPLTL. In particular, we exploit a variation of the Emerson-Lei condition \cite{EmersonL87} for handling Boolean combinations, and the possibility of translating these conditions into parity conditions (typically used for \LTL) or to fixpoint computations~\cite{DBLP:conf/fossacs/HausmannLP24}.

We show that the worst-case complexity for synthesis of \LTLfplus (resp. \PPLTLplus) is the same as for the base logics \LTLf (resp. \PPLTL), i.e., 2EXPTIME-complete (resp. EXPTIME-complete). 
The EXPTIME-complete result for synthesis in \PPLTLplus is particularly interesting because, on the one hand, it shows that the exponential gap between \PPLTL and \LTLf \cite{DeGiacomoSFR20,BonassiGFFGS23icaps} is maintained when extended to handle the full safety-progress hierarchy; and, on the other hand, it gives one a logic with the same expressive power as \LTL but for which synthesis can be solved in EXPTIME instead of 2EXPTIME. Previous efforts to achieve reactive synthesis with exponential complexity focused on proper fragments of LTL~\cite{DBLP:journals/jlap/ArtecheH24}.

We also adapt our DFA-based techniques and establish that reasoning --- satisfiability, validity, and model-checking --- for \LTLfplus (resp. \PPLTLplus) is EXPSPACE-complete (resp. PSPACE-complete). The EXPSPACE-completeness result, which may appear surprising since satisfiability and model checking for \LTL are both PSPACE-complete~\cite{DBLP:reference/mc/2018}, in fact confirms and extends a recent EXPSPACE-hardness result for model checking the fragment of \LTLfplus limited to the guarantee class~\cite{DBLP:conf/atva/BansalLTVW23}. In other words, although \LTLfplus defines infinite-trace properties using explicit reference to finite-trace properties defined in \LTLf, it provides a computational advantage for synthesis but not for reasoning. Conversely, reasoning in \PPLTL has the same cost as reasoning in \LTL.



\section{Preliminaries}

For a non-negative integer $k$, let $[k] = \{1,2,\ldots,k\}$.

\subsubsection{Trace Properties}
Let $\Sigma$ be a finite  alphabet. If $S \subseteq \Sigma$, then $S^\omega$ (resp. $S^*$) is the set of infinite (resp. finite) sequences over $\Sigma$. The empty sequence is denoted $\emptystring$. Indexing sequences starts at $0$, and we write $\tau = \tau_0 \tau_1 \cdots$. For $0 \leq i < |\tau|$, write $\trace_{\leq i}$ for the prefix $\tau_0 \cdots \tau_i$, and $\trace_{\geq i}$ for the suffix $\trace_i \trace_{i+1}  \cdots$ of $\trace$. 
%
Let $AP$ be a finite non-empty set of \emph{atoms}. We assume that $AP$ is fixed, and do not measure its size in any of the complexities. A \emph{trace} is a non-empty finite or infinite sequence $\trace$ over $\Sigma=2^{\AP}$ (valuations of atoms); in particular, the empty sequence $\lambda$ is not a trace. 
The length of $\tau$ is denoted $|\tau| \in \mathbb{N} \cup \{\infty\}$. 
A \emph{finite-trace property} is a set of finite traces.  An \emph{infinite-trace property} is a set of infinite traces. A trace \emph{satisfies} a property $Z$ if it is in $Z$.
Let $\mathfrak{L}$ be a logic for representing infinite-trace (resp. finite-trace) properties.  The set of infinite (resp. finite) traces satisfying a formula $\varphi$ of $\mathfrak{L}$ is denoted $[\varphi]$.
An infinite-trace (resp. finite-trace) property $Z$ is \emph{definable} (aka, \emph{expressible}) in $\mathfrak{L}$ iff there is some formula $\varphi$ of $\mathfrak{L}$ such that $Z = [\varphi]$.

\subsubsection{Prefix Quantifiers}
%
For a finite-trace property $Z$, define four fundamental infinite-trace properties as follows: $\safe{Z}$ (resp. $\guar{Z}$) denotes the set of infinite traces $\tau$ such that every (resp. some) non-empty finite prefix of $\tau$ satisfies $Z$;
$\recu{Z}$ (resp. $\pers{Z}$) denotes the set of infinite traces $\tau$ such that infinitely many (resp. all but finitely many, aka almost all) non-empty finite prefixes of $\tau$ satisfies $Z$.

\subsubsection{Transition Systems} \label{sec:TS}


A \emph{nondeterministic transition system} $T = (\Sigma,Q,I,\delta)$ consists of a finite \emph{input alphabet} $\Sigma$ (typically, $\Sigma = 2^{\AP}$), a finite set $Q$ of \emph{states}, a set $I \subseteq Q$ of \emph{initial states}, and a \emph{transition relation} $\delta \subseteq Q \times \Sigma \times Q$. The \emph{size of $T$} is $|Q|$, the number of its states.  We say that $T$ is \emph{total} if for every  $q,a$ there exists $q'$ such that $(q,a,q') \in \delta$. Unless stated otherwise, nondeterministic transition systems are assumed to be total. In case $I = \{\iota\}$ is a singleton and $\delta$ is functional, i.e., for every $q,a$ there is a unique $q'$ such that $(q,a,q') \in \delta$, the transition system is called \emph{deterministic} instead of \emph{nondeterministic}; in this case we write $\iota$ instead of $I$, and $\delta(s,a)=s'$ instead of $(s,a,s') \in \delta$. 
A \emph{run} (aka \emph{path}) induced by the trace $\tau$ is a sequence $\rho = \rho_0 \rho_1 \cdots$ of states, where $\rho_0 \in I$ and $\rho_{i+1} \in \delta(\rho_i,\tau_i)$ for $0 \leq i < \len{\tau}$. We simply say that \emph{$\tau$ has the run $\rho$} and that \emph{$\rho$ is a run of $\tau$}. In a nondeterministic transition system $T$, a trace may have any number of runs (including none). If a trace $\tau$ has a run in $T$, we say that $T$ \emph{generates} $\tau$. 
%
%
%
For $i \in \{1,2\}$, let $T_i = (\Sigma,Q_i,\iota_i,\delta_i)$ be  deterministic transition systems over the same alphabet. The \emph{product} $T_1 \times T_2$ is the deterministic transition system $(\Sigma,Q',\iota',\delta'$) where $Q' = Q_1 \times Q_2$, $\iota' = (\iota_1,\iota_2)$, and $\delta((q_1,q_2),a) = (\delta_1(q_1,a),\delta_2(q_2,a))$. This naturally extends to a product of $n$-many systems. 
For an infinite run $\rho$, define $\inf(\rho) \subseteq Q$ to be the set of states $q \in Q$ such that $q = \rho_i$ for infinitely many $i$.


\subsubsection{Automata on finite traces} \label{sec:DA}

%
A \emph{finite automaton} is a pair $\AUTA = (T,F)$, where $T$ is a transition system and $F \subseteq Q$ is the set of \emph{final} states. The \emph{size of $\AUTA$} is the size $T$. If $T$ is nondeterministic, then $\AUTA$ is called a \emph{nondeterministic finite automaton (NFA)}; if $T$ is deterministic, then $\AUTA$ is called a \emph{deterministic finite automaton (DFA)}.
A finite trace is \emph{accepted} by $\AUTA$ if it has a run that ends in a state of $F$. The set of finite traces accepted by $\AUTA$ is called the \emph{language of $\AUTA$}, and is denoted $L(\AUTA)$. 
Two automata $\AUTA,\AUTB$ are \emph{equivalent} if $L(\AUTA) = L(\AUTB)$. 
%
We say that $\AUTA$ is \emph{equivalent} to an \LTLf/\PPLTL\ formula $\psi$ if $L(\AUTA) = [\psi]$.

\subsubsection{Games and Synthesis} \label{sec: games}

An \emph{arena} is a deterministic transition system $D = (\Sigma,Q,\iota,\delta)$ where $\Sigma = 2^{\AP}$ and $\AP$ is partitioned into $X \cup Y$ for some sets $X,Y$ of atoms. Elements of $2^Y$ (resp. $2^X$) are called \emph{environment moves} (resp. \emph{agent} moves).
An \emph{objective} $O$ over $D$ is a subset of $Q^\omega$. Elements of $O$ are said to \emph{satisfy $O$}. A \emph{game} is a pair $G = (D,O)$.
A \emph{strategy} is a function $\sigma:(2^Y)^* \to (2^X)$ that maps sequences of environment moves to agent moves (including the empty sequence, because the agent moves first).  An \emph{outcome} of $\sigma$ is an infinite trace $\tau$ over $2^{\AP}$ such that (i) $\tau_0 \cap X = \sigma(\emptystring)$, and (ii) for $i > 0$ we have $\tau_i \cap X = \sigma(\tau_0 \cap Y \cdot \tau_1 \cap Y  \cdots \tau_{i-1} \cap Y)$. A strategy $\sigma$ is \emph{winning} if for every outcome $\tau$, the run $\rho \in Q^\omega$ induced by the trace $\tau$ in $D$ satisfies $O$. 
The following computational problem is called \emph{game solving}: given $G = (D,O)$, decide if there is a winning strategy, and if so, to return a finite representation of one. 
Intuitively, in a game, the agent moves first, the environment responds, and this repeats producing an infinite trace $\tau$. The agent is trying to ensure that the run induced by $\tau$ satisfies the objective, while the environment is trying to ensure it does not.

\subsubsection*{Emerson-Lei automata and games}
An \emph{Emerson-Lei (EL) condition} over a set $Q$ of states is a triplet $(\Gamma, \lambda, B)$, where $\Gamma$ is a finite set of \emph{labels}, $\lambda: Q \to 2^\Gamma$ is a \emph{labeling function} that assigns to a state a (possibly empty) subset of labels, and $B: 2^\Gamma \to \tf$ is a Boolean function over the set $\Gamma$ (treated as variables).
For a state $q \in Q$ and a label $l \in \lambda(q)$, we say that $q$ \emph{visits} $l$.
We will sometimes (but not always) write $B$ as a Boolean formula over the set $\Gamma$ of variables, with the usual syntax and semantics, e.g., the formula $l_1 \land l_2 \land \lnot l_3$ denotes the Boolean function that assigns $\true$ to a set $Z \subseteq \Gamma$ iff $Z$ contains $l_1$ and $l_2$ but not $l_3$.
For a sequence $\rho \in Q^\omega$, we denote by $\inf_\lambda(\rho) = \bigcup \{\lambda(q) : q \in \inf(\rho)\}$, i.e., the set of labels that are visited infinitely many times by states on $\rho$. The condition $(\Gamma, \lambda, B)$ induces the objective $O = \{\rho \in Q^\omega :  B(\inf_\lambda(\rho)) = \true\}$, i.e., $\rho \in O$ if the set of labels that are visited infinitely often satisfies $B$; in this case we say that $\rho$ \emph{satisfies the Emerson-Lei condition}. 
An \emph{Emerson-Lei automaton} is a pair $\AUTA = (T,(\Gamma, \lambda, B))$, where $T$ is a transition system, and $(\Gamma, \lambda, B)$ is an Emerson-Lei condition over the states of $T$. If $T$ is deterministic then $\AUTA$ is called a \emph{deterministic Emerson-Lei automaton (DELA)}.
An \emph{Emerson-Lei (EL) game} is a DELA $(D,(\Gamma, \lambda, B))$ where $D$ is an arena.

\begin{theorem}\cite{DBLP:conf/fossacs/HausmannLP24} \label{thm: solving EL-games}
Emerson-Lei games can be solved in time polynomial in the size of the arena and exponential in the number of labels.
\end{theorem}

\section{Linear Temporal Logics on Finite Traces}

We base our infinite-trace logics on finite-trace logics. The latter should have two features: (1)  each formula in the logic can be converted into an  equivalent DFA; (2) it must be efficiently closed under Boolean operations, i.e., one should be able to obtain the negation of a formula in polynomial time. For example, regular expressions satisfy (1) but not (2). On the other hand, Linear Temporal Logic on Finite Traces (\LTLf)~\cite{DegVa13} and Pure Past Linear Temporal Logic (\PPLTL)~\cite{LichtensteinPZ85} satisfy both (1) and (2). We now recall these.
%
%

\subsubsection{\LTLf} The syntax is given by:
  $  \varphi ::= 
        p 
    \mid 
        \neg \varphi 
    \mid 
        \varphi \land \varphi 
    \mid 
        \varphi \LTLX \varphi 
        \mid 
        \varphi \LTLU \varphi
        $,  where $p \in AP$. 
Here $\LTLX$ (``next'') and $\LTLU$ (``until'') are the \emph{future operators}.  Common abbreviations include $\LTLWX \varphi = \lnot \LTLX \lnot \varphi$ (``weak next"). 
We interpret \LTLf formulas over finite traces. Intuitively, we evaluate starting in position $0$ of the trace, $\LTLX \varphi$ says that $\varphi$ holds in the next step, and $\varphi_1 \LTLU \varphi_2$ says that $\varphi_2$ holds eventually, and $\varphi_1$ holds at every point in time until then (the formal semantics of \LTLf is in the supplement).
In \LTLf, one can predicate about both ends of the finite trace: the {first instant} of the trace by avoiding temporal operators (as in \LTL), and  the {last instant} by using the abbreviation $\last := \lnot\LTLX\true$. 

\subsubsection{\PPLTL}
The syntax is given by:
  $  \varphi ::= 
        p 
    \mid 
        \neg \varphi 
    \mid 
        \varphi \land \varphi 
    \mid \varphi \LTLY \varphi 
    \mid \varphi \LTLS \varphi
    $, where $p \in AP$.
Here $\LTLY$ (``yesterday'') and $\LTLS$ (``since'') are the \emph{past operators}, analogues of ``next" and ``until", respectively, but in the past.
Common abbreviations include  $\LTLO \varphi = \true \LTLS \varphi$ (``(at least) once in the past''), $\LTLH \varphi = \lnot \LTLO \lnot \varphi$ (``historically'' or ''hitherto''), $\textit{first} = \lnot \LTLY \true$.
%
%
We interpret \PPLTL formulas over finite traces starting at the last position of the trace (the formal semantics of \PPLTL is in the supplement).
One can also predicate about both ends of the finite trace: the {last instant} of the trace by avoiding temporal operators,  and the {first instant} by using the abbreviation $\textit{first} := \lnot
\LTLY \true$.

Obviously these logics satisfy property (2) above. For property (1), we recall:
\begin{theorem}\label{thm:logic to FA}
\begin{enumerate}
\item There is an algorithm that converts a given \LTLf formula $\varphi$ into an equivalent NFA 
of size exponential in the size of $\varphi$~\cite{DegVa15}, and also into an equivalent DFA of size double-exponential in the size of $\varphi$.

\item There is an algorithm that converts a given \PPLTL formula $\varphi$ into an equivalent \DFA of size exponential in the size of $\varphi$, i.e., of size $2^{O(|\varphi|)}$~\cite{DeGiacomoSFR20}. 
\end{enumerate}
\end{theorem}

\begin{remark} \label{ref: succinct}
Although \LTLf and \PPLTL define the same set of finite trace properties (namely, the properties definable by FOL over finite traces, or equivalently by the star-free regular expressions~\cite{GPSS80,LichtensteinPZ85}), translating \LTLf into \PPLTL, and vice-versa, is at least exponential in the worst case, see~\cite{ArtaleGGMM23} and the discussions in~\cite{DeGiacomoSFR20,BonassiGFFGS23icaps,BonassiGFFGS23ecai,GeattiMR24}. 
\end{remark}

\section{\LTLfplus and \PPLTLplus} 

Based on \LTLf (resp. \PPLTL), we define the logic \LTLfplus (resp. \PPLTLplus) for specifying infinite-trace properties.
Fix a set $\AP$ of atomic predicates.


\subsubsection{Syntax} 

The syntax of \LTLfplus (resp. \PPLTLplus) is given by the following grammar:
\[
\Psi ::= \safe{\Phi} \mid \guar{\Phi} \mid \recu{\Phi} \mid \pers{\Phi} \mid \Psi \lor \Psi \mid \Psi \land \Psi \mid \lnot \Psi
\]
where $\Phi$ are formulas in \LTLf  (resp. \PPLTL) over $AP$. We use common abbreviations, e.g., $\Psi \limp \Psi'$ for $\lnot \Psi \lor
\Psi'$.
The formulas  $\safe{\Phi}, \guar{\Phi}, \recu{\Phi}, \pers{\Phi}$ are called \emph{infinite-trace formulas}, and the formulas $\Phi$ are called \emph{finite-trace formulas}.




\subsubsection{Semantics}
For an \LTL/\PPLTL formula $\Phi$, recall that we write $[\Phi] \subseteq (2^{\AP})^*$ for the set of finite
traces that satisfy $\Phi$. For an \LTLfplus/\PPLTLplus formula $\Psi$ let $[\Psi] \subseteq
(2^{\AP})^\omega$ denote the set of infinite-trace properties, defined recursively as follows: 
\begin{itemize}
    \item $[\Psi \lor \Psi'] = [\Psi] \cup [\Psi']$, $[\Psi \land \Psi'] = [\Psi] \cap [\Psi']$, and $[\lnot \Psi] = (2^{\AP})^\omega \setminus [\Psi]$;
    \item $[\safe{\Phi}] = \safe{[\Phi]}$, $[\guar{\Phi}] = \guar{[\Phi]}$, $[\recu{\Phi}] = \recu{[\Phi]}$, and $[\pers{\Phi}] = \pers{[\Phi]}$.
\end{itemize}

In words, the Boolean operations are handled as usual, and $[\safe{\Phi}]$ (resp. $[\guar{\Phi}]$) denotes the set of infinite traces $\tau$ such that every (resp. some) non-empty finite prefix of $\tau$ satisfies $\Phi$, and $[\recu{\Phi}]$ (resp. $[\pers{\Phi}]$) denotes the set of infinite traces $\tau$ such that infinitely many (resp. all but finitely many) non-empty finite prefixes of $\tau$ are in $\Phi$. Note that $\safe{\Phi}$ and $\guar{\Phi}$ are dual since $\safe{\Phi} \equiv \lnot\guar{\lnot \Phi}$, as are $\recu{\Phi}$ and $\pers{\Phi}$ since $\recu{\Phi} \equiv
\lnot\pers{\lnot \Phi}$.
We write $\tau \models \Psi$ to mean that $\tau \in [\Psi]$.

In the terminology of \cite{DBLP:conf/podc/MannaP89}: $[\safe{\Phi}]$ is a \emph{safety} property; $[\guar{\Phi}]$ is a \emph{guarantee} property; $[\pers{\Phi}]$ is a \emph{persistence} property; $[\recu{\Phi}]$ is a \emph{recurrence} property; and $[\Psi]$ is a \emph{reactivity} property. 

\begin{remark} \label{rem: shorthands} 
To predicate about the initial state we can include
propositional formulas $\phi$ as abbreviations: for $\safe{\phi}$ in \LTLfplus, and for 
\safe{(\LTLH (\textit{first} \limp \phi))} in \PPLTLplus.
\end{remark}



\begin{remark}
    For \PPLTLplus we could use $\LTLG$ instead of $\forall$ and $\LTLF$
    instead of $\exists$, without changing the meaning of the \PPLTLplus formulas. 
    This would match the syntax used for safety-progress
    hierarchy in \cite{DBLP:conf/podc/MannaP89}.
\end{remark}



\subsubsection{Applications}
We illustrate \LTLfplus and \PPLTLplus with examples from planning.
Planning domains \cite{GeBo13}, non-Markovian domains \cite{Gabaldon11}, and safety
properties more generally, can all be expressed with formulas of the form
$\safe{\Phi}$, see also \cite{DBLP:conf/eumas/AminofGSFRZ23}. For example --- below $\phi_{x}$ are propositional formulas --- conditional effect axioms can be expressed
in \LTLfplus and \PPLTLplus by 
$\safe{\LTLG(\phi_{c} \limp (\LTLX a \limp \LTLX \phi_e))}$ and $
\safe{\LTLH(\LTLY\phi_{c} \limp (a \limp \phi_e))}$; 
Frame axioms can be expressed as
$\safe{\LTLG(\phi_{c} \land \LTLX a \limp (p \equiv \LTLX p))}$ and $\safe{\LTLH(\LTLY\phi_{c} \land a \limp (\LTLY p \equiv p))}$;
Initial conditions $\phi_{init}$ can be expressed using the shorthand in Remark~\ref{rem: shorthands};
Fairness of effects, e.g., assuming that action $a$ under condition
$\phi_c$ can effect two effects $\phi_{e_1}$ or $\phi_{e_2}$, can be expressed as 
$\recu{\LTLF (\phi_c \land \LTLX(a \land \last))} \limp \bigwedge_i \recu{\LTLF (\phi_c \land \LTLX(a \land \last \land \phi_{e_i}))}$ and 
$\recu{((\LTLY \phi_c) \land a )} \limp \bigwedge_i \recu{((\LTLY \phi_c) \land a  \land \phi_{e_i})}$.
Also, arbitrary \LTLf temporally extended goals $\Phi$~\cite{DegVa13} can be expressed using formulas of the
form $\guar{\Phi}$, e.g., eventually reaching a state where
$\phi_{g}$ holds while guaranteeing to maintain a safety condition
$\phi_{s}$ can be expressed as
$\guar{(\LTLF\phi_g \land \LTLG\phi_s) }$ and $\guar{(\phi_g\land \LTLH\phi_s)}$.\footnote{
Although in these examples the sizes of the \LTLf formulas and
corresponding \PPLTL are linearly related, this does not generalize to arbitrary formulas; see Remark~\ref{ref: succinct}.}

\subsection{Expressive Power}

%
%

The following says that the logics \PPLTLplus, \LTLfplus,  and \LTL have the same expressive power.

\begin{theorem}
    The logics \LTLfplus,  \PPLTLplus, and \LTL define the same infinite-trace properties.
\end{theorem}

\begin{proof}
    For $\mathbb{Q} \in \{\safeSymbol,\guarSymbol,\recuSymbol,\persSymbol\}$, 
    write $\mathbb{Q}\PPLTL$ for the fragment of \PPLTL consisting of formulas of the form $\mathbb{Q}\Phi$ for $\Phi$ in \PPLTL. 
Recall from~\cite{DBLP:conf/podc/MannaP89} that: a \emph{reactivity formula} is defined as a Boolean combination of formulas from $\recu{\PPLTL}$, and that every formula in $\safe{\PPLTL}, \guar{\PPLTL}, \pers{\PPLTL}$ is equivalent to a reactivity formula.\footnote{This equivalence may involve a blowup, but this is not relevant to understand expressive power.} Thus, \PPLTLplus defines the same properties as the reactivity formulas. Also, recall from~\cite{DBLP:conf/podc/MannaP89} that 
the reactivity formulas define the same infinite-trace properties as \LTL.
Finally, use the fact that \PPLTL and \LTLf define the same finite-trace properties~\cite{DeGiacomoSFR20}.
\end{proof}

The \LTLfplus (resp. \PPLTLplus) formulas in 
\emph{positive normal-form (PNF)} are given by the following grammar:
$
\Psi ::= \safe{\Phi} \mid \guar{\Phi} \mid \recu{\Phi} \mid \pers{\Phi} \mid \Psi \lor \Psi \mid \Psi \land \Psi
$
where $\Phi$ are formulas in \LTLf  (resp. \PPLTL) over $AP$.
We can convert, in linear time, a given \LTLfplus (resp.\ \PPLTLplus) formula into one that is in positive normal-form, in the usual way, by pushing negations into the finite-trace formulas $\Phi$.

\section{Reactive Synthesis}
The synthesis problem for \LTL was introduced in~\cite{PnueliR89}.
In this section we define and study the synthesis problem for \PPLTLplus and \LTLfplus. 
\begin{definition}
The \emph{$\LTLfplus$ synthesis problem (resp. \PPLTLplus synthesis problem)} asks, given $\Psi$ in \LTLfplus (resp. \PPLTLplus),  to decide if there is a strategy such that every outcome of $\sigma$ satisfies $\Psi$ (such a strategy is said to \emph{enforce} $\Psi$), and, if so, to return a finite representation of such a strategy. 
\end{definition}

We solve the \LTLfplus/\PPLTLplus synthesis problem with an automata-theoretic approach based on DFAs, in a 4-step \textbf{Synthesis Algorithm}, as follows.
%

For convenience, we begin with a formula $\Psi$ in positive normal-form. The formula $\Psi$ can be thought of as a Boolean formula $\widehat{\Psi}$ (i.e., without negations) over a set of the form $[k]$ (for some $k$); thus $\Psi$ can be formed from $\widehat{\Psi}$ by replacing, for every $i \in [k]$, every occurrence of $i$ in $\widehat{\Psi}$ by a certain infinite-trace formula of the form $\mathbb{Q}_i \Phi_i$, where $\mathbb{Q}_i \in \{\guarSymbol,\safeSymbol,\persSymbol,\recuSymbol\}$.

\subsubsection{Step 1.}
Convert each finite-trace formula $\Phi_i$ into an equivalent DFA $\AUTA_i = (D_i,F_i)$ as in Theorem~\ref{thm:logic to FA}.\footnote{One only has to do this conversion once for each finite-trace formula, even if it appears in multiple infinite-trace formulas.} Say $D_i = (\Sigma,Q_i,\iota_i,\delta_i)$. 
We assume, without loss of generality, that the initial state has no incoming transitions. 

\subsubsection{Step 2.} For each $i \in [k]$, build a transition system $D'_i = (\Sigma,Q_i,\iota_i,\delta'_i)$, a set $F'_i \subseteq Q_i$, and an objective $O_i$ over $Q_i$:
\begin{enumerate}

\item Say $\mathbb{Q}_i = \recuSymbol$. Let $\delta'_i = \delta_i$, $F'_i = F_i$, and $O_i = \{\rho : \inf(\rho) \cap F'_i \neq \emptyset\}$.

\item Say $\mathbb{Q}_i = \persSymbol$. Let $\delta'_i = \delta_i$, $F'_i = F_i$, and $O_i  = \{ \rho : \inf(\rho) \cap (Q \setminus F'_i) = \emptyset\}$.

\item Say $\mathbb{Q}_i = \safeSymbol$. If $\iota_i$ is not a final state, make it into a final state, i.e., define $F'_i = F_i \cup \{\iota_i\}$. Let $\delta'_i$ be like $\delta_i$ except that every non-final state is a sink, i.e., define $\delta'_i(q,a)$ to be $q$ if $q \not \in F'_i$, and otherwise to be $\delta_i(q,a)$. Let $O_i = \{\rho :  \inf(\rho) \cap F'_i \neq \emptyset\}$.

\item Say $\mathbb{Q}_i = \guarSymbol$. If $\iota_i$ is a final state, make it into a non-final state, i.e., define $F'_i = F_i \setminus \{\iota_i\}$. Let $\delta'_i$ be like $\delta_i$ except that every final state $q$ is a sink, i.e., define $\delta'_i(q,a)$ to be $q$ if $q \in F'_i$, and otherwise to be $\delta_i(q,a)$.  Let $O_i =  \{\rho : \inf(\rho) \cap F'_i \neq \emptyset\}$.
\end{enumerate}

We get the following from these definitions (the proof is in the supplement):\footnote{We make two remarks about the cases $\mathbb{Q} \in \{\safeSymbol,\guarSymbol\}$: by our assumption that the initial state has no incoming transitions, redefining whether the initial state is a final state or not does not change the set of non-empty strings accepted by the DFA $\AUTA_i$; we could have defined either or both these cases to have an objective like the $\persSymbol$ case (instead of the $\recuSymbol$ case) since seeing a sink is the same as seeing it infinitely often which is the same as seeing it from some point on.}

\begin{proposition} \label{prop: quantified DFA}
An infinite trace satisfies $\mathbb{Q}_i \Phi_i$ iff the run it induces in $D'_i$ satisfies the objective $O_i$.
\end{proposition}

Call a state $q_i \in Q_i$ \emph{marked} if either (a) $\mathbb{Q}_i \neq \persSymbol$ and $q_i \in F'_i$, or (b) $\mathbb{Q}_i = \persSymbol$ and 
$q_i \in Q_i \setminus F'_i$. Thus, in case $\mathbb{Q}_i \neq \persSymbol$, the objective $O_i$ says that some marked state is seen infinitely often, and in case $\mathbb{Q}_i = \persSymbol$ the objective $O_i$ says that no marked state is seen infinitely often.

\subsubsection{Step 3.} Build the product transition system $D = \prod_{i \in [k]} D'_i$. Let $Q$ be the state set of $D$. Define an Emerson-Lei condition $(\Gamma,\lambda,B)$ over $Q$ as follows: $\Gamma = [k]$ (i.e., the labels are the component numbers), $\lambda(q) = \{i \in [k] : q_i \textrm{ is marked}\}$ (i.e., a state is labeled by the numbers of those components that are in marked states), and the Boolean formula $B$ is formed from the Boolean formula $\widehat{\Psi}$ by simultaneously replacing, for every $i$ with $\mathbb{Q}_i = \persSymbol$, every occurrence of $i$  by $\lnot i$. Define the deterministic EL-automaton $\AUTA_\Psi = (D,(\Gamma,\lambda,B))$.


\begin{example}
Say $\Psi = \recu{\Phi_1} \lor \pers{\Phi_2}$. Then the label set is $\Gamma = \{1,2\}$. Let $\tau$ be an infinite trace, let $\rho$ be the induced run in $D$, and let $\rho^i$ be the induced run in $D'_i$. Then, by Proposition~\ref{prop: quantified DFA}, $\tau \models \Psi$ iff $\rho^1$ satisfies $O_1$ or $\rho^2$ satisfies $O_2$. By definition of $O_i$, this is iff $\rho^1$ sees a marked state of $D'_1$ infinitely often or it is not the case that $\rho^2$ sees a marked state of $D'_2$ infinitely often. And this is iff the set of labels $i \in \{1,2\}$ for which $\rho^i$ sees a marked state infinitely often, satisfies the Boolean formula $B := 1 \lor \lnot 2$.
\end{example}

More generally, we get the following Proposition (the proof is in the supplement):
\begin{proposition}
 \label{prop: product} 
The formula $\Psi$ is equivalent to the EL-automaton $\AUTA_\Psi$.
\end{proposition}


\subsubsection{Step 4.} Solve the Emerson-Lei game $\AUTA_\Psi$. By Theorem~\ref{thm: solving EL-games}, this can be done in time polynomial in the size of the arena and exponential in the number of labels. 

This completes the description of the algorithm.

Since, for an \LTLfplus (resp. \PPLTLplus) formula $\Psi$, the size of the arena $D$ is double-exponential (resp. single-exponential) in the size of $\Psi$, and the number of labels $k$ is linear in the size of $\Psi$, we get:
\begin{theorem} \label{thm: solving synthesis}
The \LTLfplus (resp. \PPLTLplus) synthesis problem is in $2$\exptime (resp. \exptime).
\end{theorem}

We observe that our algorithms are optimal:
\begin{theorem} \label{thm: synthesis complexity}
The $\LTLfplus$ (resp.\PPLTLplus) synthesis problem is $2$EXPTIME-complete (resp. EXPTIME-complete).
\end{theorem}


\begin{proof} 
For the lower bounds, synthesis is 2EXPTIME-hard already for the fragment of \LTLfplus consisting of 
formulas of the form $\guar{\Phi}$~\cite{DegVa15}, and synthesis is
EXPTIME-hard already for the fragment of \PPLTLplus for formulas of the form $\guar{\Phi}$~\cite{DeGiacomoSFR20}.
\end{proof}

\begin{remark}
\PPLTLplus synthesis is in \exptime (Theorem~\ref{thm: solving synthesis}), which is exponentially cheaper than \LTL synthesis, which is $2$\exptime-hard~\cite{DBLP:conf/focs/PnueliR90}. This means that \LTL is, in the worst case, exponentially more succinct than \PPLTLplus. One may be tempted to think that the fact that \PPLTLplus synthesis is cheaper than \LTL synthesis is because \LTL is always more succinct than \PPLTLplus. However, it turns out this is not true. Indeed, \PPLTLplus is, in the worst case, exponentially more succinct than \LTL~\cite[Section $3$]{DBLP:journals/eatcs/Markey03}. This leads to the interesting phenomenon of two logics with the same expressive power, incomparable succinctness, yet solving a logic problem for one of them is exponentially cheaper. This phenomenon holds for the finite-trace logics \LTLf/\PPLTL~\cite{ArtaleGGMM23}, and so the fact that \PPLTLplus synthesis is in \exptime shows that this phenomenon can be lifted to all LTL-definable infinite-trace~properties.
\end{remark}

\subsection{A new and simple algorithm for solving Emerson-Lei games}

Emerson-Lei games can be solved in time polynomial in the size of the arena and exponential in the number of labels. We mention two approaches to do this~\cite{DBLP:conf/fossacs/HausmannLP24}. One approach amounts to solving fixpoint equation systems, and has the potential to be implemented symbolically directly. Another approach first converts EL-games to parity games, and then solves the parity game. We supply a simple algorithm that takes the second approach. It uses a variation of the ``state appearance record'' construction for converting Muller automata to parity automata~\cite{GT4CS}. This yields the same complexity as in Theorem~\ref{thm: solving EL-games}.




\subsubsection{Parity automata}

Let $D$ be a deterministic transition system with state set $Q$.
A function $c:Q \to \mathbb{N}$ is called a \emph{coloring function on $D$}. The \emph{index} of $c$ is $|c(Q)|$, i.e., the number of different colors actually used by $c$.
The \emph{parity objective} $O_c$ consists of the infinite runs $\rho$ such that the largest color seen infinitely often on $\rho$ is even, i.e., $\rho \in O_c$ iff $\max \{c(q) : q \in \inf{(\rho)}\}$ is even. 
The pair $(D,c)$ is called a \emph{deterministic parity automaton (DPA)}. A \emph{parity game} is a pair $G = (D,c)$ where $D$ is an arena. Parity games can be solved in time polynomial in the size of $D$ and exponential in the index~\cite{GT4CS}; better algorithms are known, but we will not need these here.


\subsubsection{Converting EL-automata to parity automata}

We now describe the construction of a parity automaton $\AUTB = (D',c)$ equivalent to a given Emerson-Lei automaton $\AUTA = (D,(\Gamma,\lambda,B))$.
The idea is to enrich the state space of $D$ by keeping track of two pieces of information along a run: (i) the order in which labels in $\Gamma$ were last visited --- this is done by keeping a permutation $\pi$ of the elements of $\Gamma$; (ii) a way to detect which elements of $\Gamma$ are visited infinitely often --- this is done by keeping the largest position $h$ in $\pi$ that contains a label visited by the source state. Thus, the states of $\AUTB$ are of the form $(q,\pi,h)$ where $q$ is a state of $\AUTA$. Intuitively, the job of the permutations and pointers in the states of $\AUTB$ is to help the parity condition of $\AUTB$ detect if the set of labels visited infinitely often by the corresponding run of $\AUTA$ satisfies the Boolean function $B$ in the EL condition.  The way the transition relation of $\AUTB$ works is as follows: reading a symbol $a$ at a state $(q, \pi, h)$, the first coordinate (i.e., the state of $\AUTA$) is updated according to the transition relation of $\AUTA$, the permutation is updated by pushing all the labels that the source state $q$ visits to the left, and the pointer is set to the right-most position in $\pi$ that contains a label visited by $q$. Observe that it is because of the word `infinitely' in the previous sentence that the permutation and the pointer in the initial state can be chosen arbitrarily. We now supply the formal construction.

Without loss of generality, we may suppose that $\Gamma = [k]$ for some $k$. We think of a \emph{permutation} of $\Gamma$ as a string $\pi$ that consists of every character in $[k]$ exactly once. Let $\Pi$ be the set of such permutations.  Given a subset of labels $X \subseteq \Gamma$, and a permutation $\pi \in \Pi$, denote by $\pi \Lsh X$ the permutation derived from $\pi$ by pushing the elements in $X$ (without changing their relative order) all the way to the left.
We denote by $maxPos(X, \pi)$ the maximal position of an element of $X$ in $\pi$, i.e., $maxPos(X, \pi) = \max\{h \in [0,k-1] \mid \pi_h \in X\}$. To illustrate, if $k = 5$, $\pi = 04213$, and $X = \{1,4\}$, then $\pi \Lsh X = 41023$ and $maxPos(X, \pi) = 3$.

We are now ready to present the construction of the parity automaton equivalent to a given EL-automaton. Define $\AUTB = (T', c)$ with $T' = (\Sigma,Q',I',\delta')$ as follows: $Q' = Q \times \Pi \times [0,k-1]$, $I' = I \times \{\pi'\} \times \{0\}$ where $\pi' \in \Pi$ is arbitrary, $\delta'$ is defined by
$
\delta'((q, \pi, h), a) =  (\delta(q,a), \pi \Lsh \lambda(q), maxPos(\lambda(q), \pi)),
$
and the coloring function $c$ assigns to a state $s$ of $Q'$ one of the colors $[1,2k]$ as follows: if $s = (q,\pi,h)$, let $X_s \subseteq \Gamma$ be the set $\{\pi_0, \pi_1, \cdots, \pi_h\}$, and define $c(s) = 2|X_s| -1$ if $B(X_s) = \false$, and $c(s) = 2|X_s|$ if $B(X_s) = \true$.
This gives the following (the proof is in the supplement):
\begin{theorem}\label{thm:EL to parity}
There is an algorithm that converts a given deterministic Emerson-Lei automaton $\AUTA$ into an equivalent deterministic parity automaton $\AUTB$. If $\AUTA$ has $k$ labels, then $\AUTB$ has $|Q| k! k$ states and $2k$ colors.
\end{theorem}

We apply this to the constructed EL-automaton $\AUTA_\Psi = (D,(\Gamma,\lambda,B))$ (defined at the end of Step 3, above), to get an equivalent deterministic parity automaton $\AUTB = (D',c)$. In this case, the labels are the component numbers $[k]$ of the product $D = \prod_{i \in [k]} D'_i$, a state $q$ of $D$ is labeled by $i \in [k]$ if $q_i$ is a marked state of $D'_i$ (recall that the marking of a state $q_i$ simply depends on the type of the quantifier $\mathbb{Q}_i$ and whether or not $q_i$ is a final state of $D'_i$),
and a permutation records the order in which the components last visited a marked state. We remark that the construction of the final parity automaton is modular: each finite-trace formula of $\Psi$ is converted into a DFA, each of these DFAs are combined into a product (and updated according to the transition functions of the DFA), which itself is combined in a product with the set of permutations and the set of pointers of the component numbers (and updated according to which components of the current state are in final states of their respective DFAs).

\section{Reasoning} \label{sec: reasoning}

We now study other standard computational problems for our logics. Let $\mathfrak{L}$ be one of \LTLf/\PPLTL.
The \emph{$\mathfrak{L}+$ satisfiability problem} asks, given $\Psi$ in $\QeL$, to decide whether $[\Psi] \neq \emptyset$. The
\emph{$\mathfrak{L}+$ validity problem} asks, given $\Psi$ in $\QeL$, to decide whether $[\Psi] = (2^{\AP})^\omega$.  
For a nondeterministic transition system $T$, that is not assumed to be total, with input alphabet $\Sigma = 2^{\AP}$, we say that \emph{$T$ satisfies $\Psi$} if every infinite trace generated by $T$ satisfies $\Psi$.\footnote{There are many ways to define transitions systems for verification. Commonly, these are state-labeled directed graphs. For technical convenience, instead of labeling the state with $l$, we will instead label all outgoing edges with $l$, and thus we can view these as (not necessarily total) nondeterministic transition systems.}
The \emph{$\QeL$ model checking problem} asks, given $\Psi$ in $\QeL$, and a nondeterministic labeled transition system $T$, whether $T$ satisfies~$\Psi$.

\begin{theorem} \label{thm: satisfiability et al}
The \LTLfplus (resp. \PPLTLplus) satisfiability, validity, and model checking problems are EXPSPACE-complete (resp. PSPACE-complete).
\end{theorem}

The lower bounds are based on existing results~\cite{DegVa13}, or standard encodings~\cite{DBLP:conf/stoc/VardiS85,kupferman2000automata,DBLP:conf/atva/BansalLTVW23}, already for formulas of the form $\safe{\Phi}$ and $\guar{\Phi}$ (details are given in the supplement). We now present the algorithms, which give the matching upper bounds.

\subsection{Solution Techniques for Reasoning}

We provide optimal algorithms for satisfiability. These can be adapted, in the standard way, to give optimal algorithms for model-checking and validity (we do this in the supplement).  Given $\Psi$, build the equivalent EL-automaton $\AUTA_\Psi = (D,(\Gamma,\lambda,B))$ (defined at the end of Step 3, above). Recall that $|\Gamma|=k$, which is linear in $|\Psi|$. Satisfiability can be solved by checking if there is an infinite run in $D$ that satisfies the EL-condition.\footnote{For model-checking, check if there is a trace generated by $T$ that induces a run in $D$ that satisfies the EL-condition for the negation of the formula, and then complement the answer.} This can be done in time polynomial in the size of $D$ and $k$, as follows. Check if there is a state $s$ of $D$ such that (i) there is a path from the initial state to $s$, and (ii) there is a path (that uses at least one transition) from $s$ to $s$ such that the set of labels visited on this path is a satisfying assignment of $B$. 
However, since the size of $D$ is doubly-exponential (resp. exponential) in the size $|\Psi|$ of the given formula, this algorithm would run in 2EXPTIME for \LTLfplus (resp. EXPTIME for \PPLTLplus). To lower the complexities to EXPSPACE (resp. PSPACE) we observe that the search for such a state $s$ can be done nondeterministically while consulting $D$ ``on-the-fly'' using standard space-efficient techniques for finding paths in exponentially large graphs~\cite{DBLP:conf/banff/Vardi95}, and then applying Savitch's Theorem~\cite{Savitch70}. For completeness, we now outline what this amounts to in our particular case.

Instead of constructing $D$ explicitly, we nondeterministically guess the state $s$ and verify the existence of the paths satisfying (i) and (ii). For the verification, the algorithm should not store every state on the path (it would run out of space if it did this); besides storing $s$, at each step it only needs to store the current state and the next (guessed) state. Thus, since we will be using a binary encoding for the states, if there are $N$ states the algorithm uses $O(\log N)$ space. Moreover, for this procedure to work we only need to perform the following building blocks $(\dagger)$: (a) given a state $s$ of $D$, compute its labels in polynomial time in the size of $s$ (and evaluate the formula $B$ on a set of labels; recall that the size of the formula $B$ is linear in $|\Psi|$, and formula evaluation can be done in linear time), and (b) given states $s,s'$ and $a \in \Sigma$, verify in polynomial time, in the size of $s$ and $s'$, if there is a transition in $D$ labeled $a$ that takes $s$ to $s'$. Both can be done by accessing (but not building) the states and transitions of the corresponding DFAs (resp. NFAs) that make up the product $D$ (see the supplement for details).


\section{Related Work}

Besides the extensive related work discussed in the introduction, we mention that Manna and Pnueli's safety-progress hierarchy has mainly been used to understand the expressivity of fragments of \LTL. However, recent results that convert arbitrary \LTL formulas into a normal form consisting of a combination of classes in the hierarchy in single exponential time \cite{EsparzaKS20,EsparzaRS24} are providing new interest in the hierarchy, especially in the variant where classes are expressed in \LTL
\cite{ChangMP92}. 
Our work on \LTLfplus and \PPLTLplus shares this interest, though sidesteps the \LTL normalization, and focuses instead on specifying directly using finite traces, and then lifting to infinite traces through quantification on prefixes.

\section{Conclusion} 

We introduced two logics \LTLfplus/\PPLTLplus that have the same expressive power as \LTL, and showed how to apply DFA techniques to build arenas for synthesis, typical for handling the base finite-trace logics, thus side-stepping the difficulty of determinizing automata on infinite traces, typical of LTL synthesis.
Our approach only relies on the base logic being closed under negation and having conversions into DFAs. Thus, it would also work for any other base logic with these properties, e.g., linear dynamic logic on finite traces~\cite{DegVa15} or pure past linear dynamic logic~\cite{DeGiacomoSFR20}.
We thus advocate revisiting Manna and Pnueli's safety-progress hierarchy, focusing on (i) using it directly for representation, by putting emphasis on the finite-trace components of infinite properties, and (ii) solving computational problems  by exploiting DFA-based arena constructions. 

Finally, we mention that our solution to synthesis requires solving an Emerson-Lei game built from the product of the DFAs with a simple Boolean formula (that mimics the given formula). This can be implemented based on fixpoint computations~\cite{DBLP:conf/fossacs/HausmannLP24}, which has the potential for efficient symbolic implementations, or by  annotating the product with permutations of and pointers to the component numbers, to get a parity game, and then using existing parity-game solvers~\cite{DBLP:conf/tacas/Dijk18}.

\bibliography{references.bib}

\begin{thebibliography}{56}
\providecommand{\natexlab}[1]{#1}

\bibitem[{Alberto, Bienvenu, and McIlraith(2019)}]{CamachoBM19}
Alberto; Bienvenu, M.; and McIlraith, S.~A. 2019.
\newblock Towards a Unified View of {AI} Planning and Reactive Synthesis.
\newblock In \emph{{ICAPS}}, 58--67. {AAAI} Press.

\bibitem[{Aminof et~al.(2023)Aminof, {Giuseppe De Giacomo}, {Antonio Di Stasio
  }and Hugo~Francon, Rubin, and Zhu}]{DBLP:conf/eumas/AminofGSFRZ23}
Aminof, B.; {Giuseppe De Giacomo}; {Antonio Di Stasio }and Hugo~Francon; Rubin,
  S.; and Zhu, S. 2023.
\newblock {LTLf} Synthesis Under Environment Specifications for Reachability
  and Safety Properties.
\newblock In \emph{EUMAS}.

\bibitem[{Apt and Gr{\"a}del(2011)}]{GT4CS}
Apt, K.; and Gr{\"a}del, E. 2011.
\newblock \emph{Lectures in game theory for computer scientists}.
\newblock Cambridge University Press.

\bibitem[{Artale et~al.(2023)Artale, Geatti, Gigante, Mazzullo, and
  Montanari}]{ArtaleGGMM23}
Artale, A.; Geatti, L.; Gigante, N.; Mazzullo, A.; and Montanari, A. 2023.
\newblock {LTL} over Finite Words Can Be Exponentially More Succinct Than
  Pure-Past LTL, and vice versa.
\newblock In \emph{{TIME}}.

\bibitem[{Arteche and Hermo(2024)}]{DBLP:journals/jlap/ArtecheH24}
Arteche, N.; and Hermo, M. 2024.
\newblock Towards the Exact Complexity of Realizability for Safety {LTL}.
\newblock \emph{J. Log. Algebraic Methods Program.}, 141: 101002.

\bibitem[{Bacchus and Kabanza(1998)}]{DBLP:journals/amai/BacchusK98}
Bacchus, F.; and Kabanza, F. 1998.
\newblock Planning for Temporally Extended Goals.
\newblock \emph{Ann. Math. Artif. Intell.}, 22(1-2): 5--27.

\bibitem[{Bacchus and Kabanza(2000)}]{BacchusK00}
Bacchus, F.; and Kabanza, F. 2000.
\newblock Using Temporal Logics to Express Search Control Knowledge for
  Planning.
\newblock \emph{Artif. Intell.}, 116(1-2): 123--191.

\bibitem[{Baier, Fritz, and McIlraith(2007)}]{BaierFM07}
Baier, J.~A.; Fritz, C.; and McIlraith, S.~A. 2007.
\newblock Exploiting Procedural Domain Control Knowledge in State-of-the-Art
  Planners.
\newblock In \emph{{ICAPS}}.

\bibitem[{Baier and McIlraith(2006)}]{BaierM06}
Baier, J.~A.; and McIlraith, S.~A. 2006.
\newblock Planning with First-Order Temporally Extended Goals using Heuristic
  Search.
\newblock In \emph{{AAAI}}.

\bibitem[{Bansal et~al.(2020)Bansal, Li, Tabajara, and
  Vardi}]{bansal2020hybrid}
Bansal, S.; Li, Y.; Tabajara, L.~M.; and Vardi, M.~Y. 2020.
\newblock {Hybrid Compositional Reasoning for Reactive Synthesis from
  Finite-Horizon Specifications}.
\newblock In \emph{{AAAI}}.

\bibitem[{Bansal et~al.(2023)Bansal, Li, Tabajara, Vardi, and
  Wells}]{DBLP:conf/atva/BansalLTVW23}
Bansal, S.; Li, Y.; Tabajara, L.~M.; Vardi, M.~Y.; and Wells, A.~M. 2023.
\newblock Model Checking Strategies from Synthesis over Finite Traces.
\newblock In \emph{{ATVA}}.

\bibitem[{Bonassi et~al.(2023{\natexlab{a}})Bonassi, {De Giacomo}, Favorito,
  Fuggitti, Gerevini, and Scala}]{BonassiGFFGS23ecai}
Bonassi, L.; {De Giacomo}, G.; Favorito, M.; Fuggitti, F.; Gerevini, A.; and
  Scala, E. 2023{\natexlab{a}}.
\newblock {FOND} Planning for Pure-Past Linear Temporal Logic Goals.
\newblock In \emph{{ECAI}}.

\bibitem[{Bonassi et~al.(2023{\natexlab{b}})Bonassi, {De Giacomo}, Favorito,
  Fuggitti, Gerevini, and Scala}]{BonassiGFFGS23icaps}
Bonassi, L.; {De Giacomo}, G.; Favorito, M.; Fuggitti, F.; Gerevini, A.; and
  Scala, E. 2023{\natexlab{b}}.
\newblock Planning for Temporally Extended Goals in Pure-Past Linear Temporal
  Logic.
\newblock In \emph{{ICAPS}}.

\bibitem[{Bonassi et~al.(2024)Bonassi, {De Giacomo}, Gerevini, and
  Scala}]{BonassiDGS24ecai}
Bonassi, L.; {De Giacomo}, G.; Gerevini, A.; and Scala, E. 2024.
\newblock Shielded {FOND}: Planning With Safety Constraints in Pure-Past Linear
  Temporal Logic.
\newblock In \emph{{ECAI}}.

\bibitem[{Calvanese, {De Giacomo}, and Vardi(2002)}]{CalvaneseGV02}
Calvanese, D.; {De Giacomo}, G.; and Vardi, M.~Y. 2002.
\newblock Reasoning about Actions and Planning in {LTL} Action Theories.
\newblock In \emph{{KR}}.

\bibitem[{Chang, Manna, and Pnueli(1992)}]{ChangMP92}
Chang, E.~Y.; Manna, Z.; and Pnueli, A. 1992.
\newblock Characterization of Temporal Property Classes.
\newblock In \emph{{ICALP}}.

\bibitem[{Cimatti et~al.(2020)Cimatti, Geatti, Gigante, Montanari, and
  Tonetta}]{CimattiGGMT20}
Cimatti, A.; Geatti, L.; Gigante, N.; Montanari, A.; and Tonetta, S. 2020.
\newblock Reactive Synthesis from Extended Bounded Response {LTL}
  Specifications.
\newblock In \emph{{FMCAD}}.

\bibitem[{Cimatti et~al.(2003)Cimatti, Pistore, Roveri, and
  Traverso}]{Cimatti03}
Cimatti, A.; Pistore, M.; Roveri, M.; and Traverso, P. 2003.
\newblock Weak, Strong, and Strong Cyclic Planning via Symbolic Model Checking.
\newblock \emph{Artif. Intell.}, 1--2(147).

\bibitem[{Clarke et~al.(2018)Clarke, Henzinger, Veith, and
  Bloem}]{DBLP:reference/mc/2018}
Clarke, E.~M.; Henzinger, T.~A.; Veith, H.; and Bloem, R., eds. 2018.
\newblock \emph{Handbook of Model Checking}.
\newblock Springer.
\newblock ISBN 978-3-319-10574-1.

\bibitem[{{De Giacomo} et~al.(2020){De Giacomo}, {Di Stasio}, Fuggitti, and
  Rubin}]{DeGiacomoSFR20}
{De Giacomo}, G.; {Di Stasio}, A.; Fuggitti, F.; and Rubin, S. 2020.
\newblock Pure-Past Linear Temporal and Dynamic Logic on Finite Traces.
\newblock In \emph{{IJCAI}}.

\bibitem[{{De Giacomo} and Favorito(2021)}]{DF2021}
{De Giacomo}, G.; and Favorito, M. 2021.
\newblock Compositional Approach to Translate {LTL}$_f$/{LDL}$_f$ into
  Deterministic Finite Automata.
\newblock In \emph{{ICAPS}}.

\bibitem[{{De Giacomo} et~al.(2022){De Giacomo}, Favorito, Li, Vardi, Xiao, and
  Zhu}]{DeGiacomoFLVX022}
{De Giacomo}, G.; Favorito, M.; Li, J.; Vardi, M.~Y.; Xiao, S.; and Zhu, S.
  2022.
\newblock LTL$_f$ Synthesis as {AND-OR} Graph Search: Knowledge Compilation at
  Work.
\newblock In \emph{{IJCAI}}.

\bibitem[{{De Giacomo} and Rubin(2018)}]{DR-IJCAI18}
{De Giacomo}, G.; and Rubin, S. 2018.
\newblock Automata-Theoretic Foundations of FOND Planning for
  {LTL}$_f$/{LDL}$_f$ Goals.
\newblock In \emph{{IJCAI}}.

\bibitem[{De~Giacomo and Vardi(2013)}]{DegVa13}
De~Giacomo, G.; and Vardi, M.~Y. 2013.
\newblock Linear Temporal Logic and Linear Dynamic Logic on Finite Traces.
\newblock In \emph{IJCAI}.

\bibitem[{De~Giacomo and Vardi(2015)}]{DegVa15}
De~Giacomo, G.; and Vardi, M.~Y. 2015.
\newblock Synthesis for {LTL} and {LDL} on Finite Traces.
\newblock In \emph{IJCAI}.

\bibitem[{Ehlers et~al.(2017)Ehlers, Lafortune, Tripakis, and
  Vardi}]{EhlersLTV17}
Ehlers, R.; Lafortune, S.; Tripakis, S.; and Vardi, M.~Y. 2017.
\newblock Supervisory Control and Reactive Synthesis: a Comparative
  Introduction.
\newblock \emph{Discret. Event Dyn. Syst.}, 27(2).

\bibitem[{Emerson and Lei(1987)}]{EmersonL87}
Emerson, E.~A.; and Lei, C. 1987.
\newblock Modalities for Model Checking: Branching Time Logic Strikes Back.
\newblock \emph{Sci. Comput. Program.}, 8(3): 275--306.

\bibitem[{Esparza, Kret{\'{\i}}nsk{\'{y}}, and Sickert(2020)}]{EsparzaKS20}
Esparza, J.; Kret{\'{\i}}nsk{\'{y}}, J.; and Sickert, S. 2020.
\newblock A Unified Translation of Linear Temporal Logic to
  {\(\omega\)}-Automata.
\newblock \emph{J. {ACM}}, 67(6): 33:1--33:61.

\bibitem[{Esparza, Rubio, and Sickert(2024)}]{EsparzaRS24}
Esparza, J.; Rubio, R.; and Sickert, S. 2024.
\newblock Efficient Normalization of Linear Temporal Logic.
\newblock \emph{J. ACM}, 71(2).

\bibitem[{Fijalkow et~al.(2023)Fijalkow, Bertrand, Bouyer-Decitre, Brenguier,
  Carayol, Fearnley, Gimbert, Horn, Ibsen-Jensen, Markey, Monmege, Novotný,
  Randour, Sankur, Schmitz, Serre, and Skomra}]{fijalkow2023games}
Fijalkow, N.; Bertrand, N.; Bouyer-Decitre, P.; Brenguier, R.; Carayol, A.;
  Fearnley, J.; Gimbert, H.; Horn, F.; Ibsen-Jensen, R.; Markey, N.; Monmege,
  B.; Novotný, P.; Randour, M.; Sankur, O.; Schmitz, S.; Serre, O.; and
  Skomra, M. 2023.
\newblock Games on Graphs.
\newblock \emph{arXiv}, 2305.10546.

\bibitem[{Finkbeiner(2016)}]{finkbeiner2016synthesis}
Finkbeiner, B. 2016.
\newblock Synthesis of Reactive Systems.
\newblock \emph{Dependable Softw. Syst. Eng.}, 45: 72--98.

\bibitem[{Gabaldon(2011)}]{Gabaldon11}
Gabaldon, A. 2011.
\newblock Non-{Markovian} Control in the Situation Calculus.
\newblock \emph{Artif. Intell.}, 175(1).

\bibitem[{Gabbay et~al.(1980)Gabbay, Pnueli, Shelah, and Stavi}]{GPSS80}
Gabbay, D.; Pnueli, A.; Shelah, S.; and Stavi, J. 1980.
\newblock On the Temporal Analysis of Fairness.
\newblock In \emph{{POPL}}.

\bibitem[{Geatti, Montali, and Rivkin(2024)}]{GeattiMR24}
Geatti, L.; Montali, M.; and Rivkin, A. 2024.
\newblock Foundations of Reactive Synthesis for Declarative Process
  Specifications.
\newblock In \emph{{AAAI}}.

\bibitem[{Geffner and Bonet(2013)}]{GeBo13}
Geffner, H.; and Bonet, B. 2013.
\newblock \emph{A Coincise Introduction to Models and Methods for Automated
  Planning}.
\newblock Morgan \& Claypool.

\bibitem[{Gerevini et~al.(2009)Gerevini, Haslum, Long, Saetti, and
  Dimopoulos}]{GereviniHLSD09}
Gerevini, A.; Haslum, P.; Long, D.; Saetti, A.; and Dimopoulos, Y. 2009.
\newblock Deterministic Planning in the Fifth International Planning
  Competition: {PDDL3} and Experimental Evaluation of the Planners.
\newblock \emph{Artif. Intell.}, 173(5-6): 619--668.

\bibitem[{Hausmann, Lehaut, and
  Piterman(2024)}]{DBLP:conf/fossacs/HausmannLP24}
Hausmann, D.; Lehaut, M.; and Piterman, N. 2024.
\newblock Symbolic Solution of Emerson-Lei Games for Reactive Synthesis.
\newblock In \emph{{FoSSaCS}}.

\bibitem[{Kamp(1968)}]{Kamp}
Kamp, H. 1968.
\newblock \emph{Tense logic and the theory of linear order}.
\newblock Ph.D. thesis, UCLA.

\bibitem[{Kupferman and Vardi(2000)}]{kupferman2000automata}
Kupferman, O.; and Vardi, M.~Y. 2000.
\newblock An Automata-theoretic Approach to Modular Model Checking.
\newblock \emph{{TOPLAS}}, 22(1).

\bibitem[{Lichtenstein, Pnueli, and Zuck(1985)}]{LichtensteinPZ85}
Lichtenstein, O.; Pnueli, A.; and Zuck, L.~D. 1985.
\newblock The Glory of the Past.
\newblock In \emph{Logic of Programs}, 196--218.

\bibitem[{Manna and Pnueli(1990)}]{DBLP:conf/podc/MannaP89}
Manna, Z.; and Pnueli, A. 1990.
\newblock A Hierarchy of Temporal Properties.
\newblock In \emph{{PODC}}.

\bibitem[{Manna and Pnueli(1992)}]{MannaPnueli92}
Manna, Z.; and Pnueli, A. 1992.
\newblock \emph{The temporal logic of reactive and concurrent systems -
  specification}.
\newblock Springer.

\bibitem[{Manna and Pnueli(1995)}]{MannaPnueli95}
Manna, Z.; and Pnueli, A. 1995.
\newblock \emph{Temporal verification of reactive systems - safety}.
\newblock Springer.

\bibitem[{Manna and Pnueli(2010)}]{MannaPnueli10}
Manna, Z.; and Pnueli, A. 2010.
\newblock Temporal Verification of Reactive Systems: Response.
\newblock In \emph{Essays in Memory of Amir Pnueli}, volume 6200 of
  \emph{Lecture Notes in Computer Science}, 279--361.

\bibitem[{Markey(2003)}]{DBLP:journals/eatcs/Markey03}
Markey, N. 2003.
\newblock Temporal Logic with Past is Exponentially More Succinct, Concurrency
  Column.
\newblock \emph{Bull. {EATCS}}, 79: 122--128.

\bibitem[{Piterman and Pnueli(2018)}]{PitermanP18}
Piterman, N.; and Pnueli, A. 2018.
\newblock Temporal Logic and Fair Discrete Systems.
\newblock In \emph{Handbook of Model Checking}, 27--73. Springer.

\bibitem[{Pnueli(1977)}]{Pnueli77}
Pnueli, A. 1977.
\newblock The Temporal Logic of Programs.
\newblock In \emph{FOCS}.

\bibitem[{Pnueli and Rosner(1989)}]{PnueliR89}
Pnueli, A.; and Rosner, R. 1989.
\newblock On the Synthesis of a Reactive Module.
\newblock In \emph{POPL}.

\bibitem[{Pnueli and Rosner(1990)}]{DBLP:conf/focs/PnueliR90}
Pnueli, A.; and Rosner, R. 1990.
\newblock Distributed Reactive Systems Are Hard to Synthesize.
\newblock In \emph{{FOCS}}.

\bibitem[{Rabin and Scott(1959)}]{RaSc59}
Rabin, M.~O.; and Scott, D. 1959.
\newblock Finite automata and their Decision Problems.
\newblock \emph{IBM J. Res. Dev.}, 3: 114--125.

\bibitem[{Savitch(1970)}]{Savitch70}
Savitch, W.~J. 1970.
\newblock Relationships Between Nondeterministic and Deterministic Tape
  Complexities.
\newblock \emph{J. Comput. Syst. Sci.}, 4(2): 177--192.

\bibitem[{van Dijk(2018)}]{DBLP:conf/tacas/Dijk18}
van Dijk, T. 2018.
\newblock Oink: An Implementation and Evaluation of Modern Parity Game Solvers.
\newblock In \emph{{TACAS}}.

\bibitem[{Vardi(1995)}]{DBLP:conf/banff/Vardi95}
Vardi, M.~Y. 1995.
\newblock An Automata-Theoretic Approach to Linear Temporal Logic.
\newblock In Moller, F.; and Birtwistle, G.~M., eds., \emph{Logics for
  Concurrency - Structure versus Automata}, volume 1043 of \emph{LNCS},
  238--266. Springer.

\bibitem[{Vardi(2007)}]{Vardi07}
Vardi, M.~Y. 2007.
\newblock The {B}{\"{u}}chi Complementation Saga.
\newblock In \emph{{STACS}}. Springer.

\bibitem[{Vardi and Stockmeyer(1985)}]{DBLP:conf/stoc/VardiS85}
Vardi, M.~Y.; and Stockmeyer, L.~J. 1985.
\newblock Improved Upper and Lower Bounds for Modal Logics of Programs:
  Preliminary Report.
\newblock In \emph{{STOC}}.

\bibitem[{Zhu et~al.(2017)Zhu, Tabajara, Li, Pu, and Vardi}]{ZTLPV17}
Zhu, S.; Tabajara, L.~M.; Li, J.; Pu, G.; and Vardi, M.~Y. 2017.
\newblock {Symbolic {LTL$_f$} Synthesis}.
\newblock In \emph{{IJCAI}}.

\end{thebibliography}

\clearpage 


\section{Supplement to ``Preliminaries''}

\subsubsection{Emerson-Lei games and automata}
Deterministic Emerson-Lei automata (DELA) are closed under Boolean operations, with very simple constructions, and trivial proofs.

\begin{proposition} \label{prop:DELA closure}
\begin{enumerate}
\item Given a DELA $\AUTA = (T,(\Gamma, \lambda, B))$, the DELA $\AUTB = (T, (\Gamma, \lambda, \lnot B)$ accepts the complement of $L(\AUTA)$.
\item  Given DELA $\AUTA = (T_1, (\Gamma_1, \lambda_1, B_1))$ and $\AUTB = (T_2, (\Gamma_2, \lambda_2, B_2))$, with $\Gamma_1$ and $\Gamma_2$ being disjoint, let $T$ be the product of $T_1$ and $T_2$. We have that the DELA $\AUTC = (T, (\Gamma_1 \cup \Gamma_2, \lambda, B))$, where $\lambda(p,q) = \lambda_1(p) \cup \lambda_2(q)$, accepts the language $L(\AUTA) \cup L(\AUTB)$ if we take $B = B_1 \lor B_2$, and the language $L(\AUTA) \cap L(\AUTB)$ if we take $B = B_1 \land B_2$.
\end{enumerate}
\end{proposition}

Note that a Muller condition $F \subseteq 2^Q$ --- which is satisfied by those runs $\rho$ such that $\inf(\rho) \in F$ --- is equivalent to the Emerson-Lei condition $(Q, \lambda, B)$, where $\lambda(q) = \{q\}$, and $B$ is the characteristic function of $F$. This shows that every $\omega$-regular language, and thus every \LTL-definable property in particular, is the language of some DELA.

  

\subsubsection{\LTL}
Linear Temporal Logic (LTL) \cite{Pnueli77} is one of the most used logic for specifying temporal properties. Its syntax is:
%
\[ \varphi ::= 
        p 
    \mid 
        \neg \varphi 
    \mid 
        \varphi \land \varphi 
    \mid 
        \varphi \LTLX \varphi 
    \mid 
        \varphi \LTLU \varphi
\]
where $p \in AP$. Here, $\LTLX$ (``next") and $\LTLU$ (``until") are the \emph{future operators}.
Common abbreviations are also used, such as $\varphi_1 \lor \varphi_2 = \lnot (\lnot \varphi_1 \land \lnot \varphi_2$), $\true = p \lor \lnot p$, $\LTLF \varphi = \true \LTLU \varphi$ (``eventually"), $\LTLG \varphi = \lnot \LTLF \lnot \varphi$ (``always").

\LTL formulas are interpreted over infinite traces. Let $\tau$ be a infinite trace. The positions in $\tau$ are indexed by $i$ with $0 \leq i $. For such an $i$, define $\tau,i \models \varphi$ (read "$\varphi$ satisfies $\tau$ at position $i$") inductively:
\begin{itemize}   
    \item $\tau,i \models p$ iff $p \in \tau_i$
    
    \item $\tau,i \models \neg \varphi$ iff $\tau,i \not \models \varphi$
    
    \item $\tau,i \models \varphi_1 \land \varphi_2$ iff $\tau,i \models \varphi_1$ and $\tau,i \models \varphi_2$
    
    \item $\tau,i \models \LTLX \varphi$ iff $\tau, i+1 \models \varphi$
    
    \item 	$\tau,i \models \varphi_1 \LTLU \varphi_2$ iff $\exists k$ with $i \leq k$ such that $(\tau, k \models \varphi_2$ and $\forall j$ with $i \leq j < k$ we have $\tau, j \models \varphi_1)$.
\end{itemize}     
For a formula $\varphi$ of \LTL, let $[\varphi]$ denote the set of infinite traces $\tau$ such that $\tau,0 \models \varphi$. In other words, formulas of \LTL are evaluated at the start of the given infinite trace.

\section{Supplement to ``Linear Temporal Logics on Finite Traces''}

\subsubsection{Semantics of \LTLf} Let $\tau$ be a finite trace. The positions in $\tau$ are indexed by $i$ with $0 \leq i < |\tau|$. Thus, e.g., $i+1 < |\tau|$ means that $i+1$ is a position of $\tau$.
For such an $i$, define $\tau,i \models \varphi$ (read "$\varphi$ satisfies $\tau$ at position $i$") inductively:
\begin{itemize}   
    \item $\tau,i \models p$ iff $p \in \tau_i$
    
    \item $\tau,i \models \neg \varphi$ iff $\tau,i \not \models \varphi$
    
    \item $\tau,i \models \varphi_1 \land \varphi_2$ iff $\tau,i \models \varphi_1$ and $\tau,i \models \varphi_2$
    
    \item $\tau,i \models \LTLX \varphi$ iff $i+1<|\tau|$ and $\tau, i+1 \models \varphi$
    
    \item 	$\tau,i \models \varphi_1 \LTLU \varphi_2$ iff $\exists k$ with $i \leq k < |\tau|$ such that $(\tau, k \models \varphi_2$ and $\forall j$ with $i \leq j < k$ we have $\tau, j \models \varphi_1)$.
\end{itemize}     
For a formula $\varphi$ of \LTLf, let $[\varphi]$ denote the set of non-empty finite traces $\tau$ such that $\tau,0 \models \varphi$. In other words, formulas of \LTLf are evaluated at the start of the given finite trace. 

\subsubsection{Semantics of \PPLTL}

For a trace $\tau$ and a position $i$, we define $\tau,i \models \varphi$ (read "$\varphi$ satisfies $\tau$ at position $i$") inductively (analogously to \LTLf):
\begin{itemize}   
    \item $\tau,i \models p$ iff $p \in \tau_i$
    
    \item $\tau,i \models \neg \varphi$ iff $\tau,i \not \models \varphi$
    
    \item $\tau,i \models \varphi_1 \land \varphi_2$ iff $\tau,i \models \varphi_1$ and $\tau,i \models \varphi_2$

    \item $\tau,i \models \LTLY \varphi$ iff $0<i$ and $\tau, i-1 \models \varphi$
        
    \item 	$\tau,i \models \varphi_1 \LTLS \varphi_2$ iff $\exists k$ with $0 \leq k \leq i$ such that $(\tau, k \models \varphi_2$ and $\forall j$ with $k < j \leq i$ we have $\tau, j \models \varphi_1)$.
\end{itemize}   
For a formula $\varphi$ of \PPLTL, let $[\varphi]$ denote the set of non-empty finite traces $\tau$ such that, if $n = |\tau|$, then $\tau,n-1 \models \varphi$. In other words, formulas of \PPLTL are evaluated starting at the end of the given finite trace. 

\subsubsection{Converting to DFAs}
We make some remarks about the conversions from these logics to automata:

\begin{remark} \label{rem: LTLf to DFA}
The conversion from \LTLf to NFA can be done by first converting, in linear time, the \LTLf formula into an alternating finite automaton (AFA), and then applying a simple subset construction that converts an AFA into an equivalent NFA~\cite{DegVa13}. Then, transforming an NFA into an equivalent DFA uses a simple subset construction~\cite{RaSc59}, as follows. Given an NFA $(T,F)$ with $T = (\Sigma,Q,I,\delta)$, define a DFA $(T',F')$ with $T' = (\Sigma,Q',\iota',\delta')$ where $Q' = 2^Q, \iota' = I$, define $\delta'(X,a) = \bigcup_{s \in X} \delta(s,a)$, and $F' = \{X \subseteq Q : X \cap F \neq \emptyset\}$.
Although there are exponentially many states in the DFA, each state can be stored using $Q$-many bits, i.e., in linear size. Moreover, given two states $X,Y$ of the DFA and an input letter $a$, one can decide in polynomial time if $\delta'(X,a) = Y$.
\end{remark}

\begin{remark} \label{rem: PP to DFA}
We briefly outline the conversion from \PPLTL to \DFA, cf.~\cite{BonassiGFFGS23icaps}. 
Given a \PPLTL formula $\Phi$ over $\AP$, let $S$ be the set of subformulas of $\varphi$, and build a DFA 
$(\Sigma,Q,\iota,F)$ with $\Sigma = 2^{\AP}$, $Q = 2^S \cup \{\iota\}$, and $F = \{X \subseteq S : \varphi \in X\}$. Define $\delta(\iota,v)$ to consists of the set of subformulas of $\varphi$ that are true on the word $v$ of length $1$.
Define $\delta(X,v) = Y$ to consist of subformulas $\psi$ of $\Phi$ chosen recursively as follows. If $\psi$ is an atom, put it into $Y$ if it is in the input $v$; put a conjunction $\psi = \psi_1 \land \psi_2$ in $Y$ if both $\psi_1,\psi_2$ are put into $Y$; put $\lnot \psi$ into $Y$ if $\psi$ is not put into $Y$; put $\LTLY \psi$ in $Y$ if $\psi \in X$; put $\psi = \psi_1 \LTLS \psi_2$ in $Y$ if either $\psi_2$ is put in $Y$, or $\psi_1$ is put in $Y$ and $\psi_1 \LTLS \psi_2 \in X$. 
Intuitively, the DFA remembers, in its state, the set of subformulas that hold on the word read so far, and after reading the next input letter, it updates this set. The reason this is possible for \PPLTL is that the truth of a formula $\psi$ on a word $u_0 u_1 \cdots u_n$ depends only on $u_n$ and on the truth of its subformulas on the word $u_0 \cdots u_{n-1}$\cite{DeGiacomoSFR20}; note that this property fails for \LTLf.
\end{remark}

\section{Supplement to ``Synthesis''}

\subsection{Proof of Proposition~\ref{prop: quantified DFA}}
    Let $\tau$ be an infinite trace and let $\rho$ (resp. $\rho'$) be the run in $D_i$ (resp. $D'_i$) induced by $\tau$. Recall that the empty string is not a trace. A non-empty prefix of $\tau$ corresponds to a prefix of its run of length $>1$. 

    We look at each case separately. 
    
    \begin{enumerate}
    \item $\tau$ satisfies $\recu{\Phi_i}$ iff (by Definition of $\recuSymbol$) infinitely many prefixes of $\tau$ satisfy $\Phi_i$ iff (by Theorem~\ref{thm:logic to FA}) infinitely many prefixes of $\rho$ end in $F_i$ iff $\inf(\rho) \cap F_i \neq \emptyset$ iff (since $\rho' = \rho$ and $F'_i = F_i$) $\rho'$ satisfies $O_i$.
    
    \item $\tau$ satisfies $\pers{\Phi_i}$ iff (by Definition of $\persSymbol$) almost all prefixes of $\tau$ satisfy $\Phi_i$ iff (by Theorem~\ref{thm:logic to FA}) almost all prefixes of $\rho$ end in $F_i$ iff it is not the case that infinitely many prefixes of $\rho$ end in $Q_i \setminus F_i$ iff $\inf(\rho) \cap (Q \setminus F_i) = \emptyset$ iff  (since $\rho' = \rho$ and $F'_i = F_i$) $\rho'$ satisfies $O_i$.
    
    \item $\tau$ satisfies $\safe{\Phi_i}$ iff (by Definition of $\safeSymbol$) all non-empty prefixes of $\tau$ satisfy $\Phi_i$ 
    iff (by Theorem~\ref{thm:logic to FA}) $\rho_n \in F_i$ for all $n > 0$
    iff (since the transitions from the initial state and final states in $D_i$ are unchanged in $D'_i$) $\rho'_n \in F'_i$ for all $n > 0$
    iff $\rho'_n \in F'_i$ for infinitely many $n$ (since all non-final states are sinks)
    iff $\inf(\rho') \cap F'_i \neq \emptyset$ iff $\rho'$ satisfies~$O_i$.
    
    \item $\tau$ satisfies $\guar{\Phi_i}$ iff (by Definition of $\guarSymbol$) some non-empty prefix of $\tau$ satisfies $\Phi_i$ 
    iff (by Theorem~\ref{thm:logic to FA}) $\rho_n \in F_i$ for some $n > 0$
    iff (since the transitions from the initial state and non-final states in $D_i$ are unchanged in $D'_i$, and taking the least such $n$) $\rho'_n \in F'_i$ for some $n > 0$
    iff $\rho'_n \in F'_i$ for infinitely many $n$ (since all final states are sinks)
    iff $\inf(\rho') \cap F'_i \neq \emptyset$ iff $\rho'$ satisfies $O_i$.
    \end{enumerate}
    This completes the proof.

\subsection{Proof of Proposition~\ref{prop: product}}
In Step 2 of the Synthesis Algorithm, for $i \in [k]$, the objective $O_i$ is induced by the following EL-condition $(\Gamma_i,\lambda_i,B_i)$ over $Q_i$. Let $\Gamma = \{i\}$, i.e., the only label is $i$. For $q_i \in Q_i$, let $\lambda_i(q_i) = \{i\}$ if $q_i$ is marked, and $\lambda_i(q_i) = \{\}$ otherwise. Then, for a run $\rho$ in $D'_i$, we have that $\inf_{\lambda_i}(\rho) = \{i\}$ iff some state seen infinitely often in $\rho$ is marked (and otherwise $\inf_{\lambda_i}(\rho) = \{\}$). We consider the two cases:
\begin{enumerate}
\item Say $\mathbb{Q}_i \neq \persSymbol$. Recall that $q_i$ is marked iff $q_i \in F'_i$. Thus, $\rho \in O_i$ iff $\inf(\rho) \cap F'_i \neq \emptyset$ iff $\inf_{\lambda_i}(\rho) = \{i\}$ iff $\rho$ satisfies the EL-condition  $(\Gamma_i,\lambda_i,B_i)$ where $B_i$ is the Boolean formula $i$.

\item Say $\mathbb{Q}_i = \persSymbol$. Recall that $q_i$ is marked iff $q_i \in Q_i \setminus F'_i$. Thus, $\rho \in O_i$ iff $\inf(\rho) \cap (Q_i \setminus F'_i) = \emptyset$ iff $\inf_{\lambda_i}(\rho) = \{\}$ iff $\rho$ satisfies the EL-condition  $(\Gamma_i,\lambda_i,B_i)$ where $B_i$ is the Boolean formula $\lnot i$.
\end{enumerate} 

By Proposition~\ref{prop: quantified DFA}, $\Psi$ is equivalent to the Boolean combination $\widehat{\Psi}$ of EL-automata $\AUTA_i := (D'_i,(\Gamma_i,\lambda_i,B_i))$, i.e., replace in $\Psi$ every $\mathbb{Q}_i \Phi_i$ by $\AUTA_i$.
Applying Proposition~\ref{prop:DELA closure} to this Boolean combination results in an EL-automaton whose transition system is the product $D$, and whose EL-condition is $(\Gamma,\lambda,B)$ where $\Gamma = [k]$, $\lambda((q_1,\cdots,q_n)) = \bigcup_i \lambda_i(q_i)$, and $B$ is formed from $\widehat{\Psi}$ by simultaneously replacing, for every $i$ with $\mathbb{Q}_i = \persSymbol$, every occurrence of $i$ by $\lnot i$. Now notice that this is exactly the definition of $\AUTA_\Psi$ in Step 3 of the Synthesis Algorithm.

\subsection{Proof of Theorem~\ref{thm:EL to parity}}

Let $\AUTA = (T,(\Gamma,\lambda,B))$ be an Emerson-Lei automaton. We will need some definitions before constructing the equivalent parity automaton $\AUTB$.

Denote by $\Pi$ the set of permutations of $\Gamma$, i.e., the set of all bijections $\pi: [0,k-1] \to \Gamma$. We will find it convenient to think of such permutations as strings (without repetitions) of length $k$ over the alphabet $\Gamma$ (although in this proof we will not use the string notation $\pi_i$ to refer to the $i$th character of $\pi$, but reserve this notation for another use).

Given a subset of labels $X \subseteq \Gamma$, and a permutation $\pi \in \Pi$, denote by $\pi \Lsh X$ the permutation derived from $\pi$ by pushing the elements in $X$ (without changing their relative order) all the way to the left.
We denote by $maxPos(X, \pi)$ the maximal position of an element of $X$ in $\pi$, i.e., $maxPos(X, \pi) = \max\{h \in [0,k-1] \mid \pi(h) \in X\}$. To illustrate, if $\pi = 04213$, and $X = \{1,4\}$, then $\pi \Lsh X = 41023$ and $maxPos(X, \pi) = 3$.

We formally define $\pi \Lsh X$ as the unique permutation $\pi'$ satisfying for every $x,y \in \Gamma$ that: if $x \in \Gamma$ and $y \not \in \Gamma$ then $\pi'^{-1}(x) < \pi'^{-1}(y)$, and if $x,y \in \Gamma$ or $x,y \not \in \Gamma$ then $\pi'^{-1}(x) < \pi'^{-1}(y)$ iff $\pi^{-1}(x) < \pi^{-1}(y)$.

For a non-empty $X \subseteq \Gamma$, we say that a permutation $\pi$ \emph{displays} $X$ if $X = \{\pi(0), \pi(1), \ldots, \pi(|X|-1)\}$, i.e., if (when viewed as a string) the elements of $X$ appear on $\pi$ to the left of the elements of $\Gamma \setminus X$. Finally, for $h \in [0,k-1]$, we say that $(\pi, h)$ \emph{separates} $X$, if $\pi$ displays $X$ and $h = |X| -1$, i.e., if the elements of $X$ occupy exactly the positions $0$ to $h$. 
Observe that while $\pi$ displays $k$ different sets, the pair $(\pi, h)$ separates exactly one set. 
For example, if $\pi = 04213$ then $\pi$ displays the sets $\{0\}, \{0,4\}, \{0,2,4\}, \{0,1,2,4\}, \{0,1,2,3,4\}$, and $(\pi, 2)$ separates the set 
$\{0,2,4\}$.
%

We are now ready to present the construction of a parity automaton equivalent to a given EL-automaton. Define the DPA $\AUTB = (T', c)$ with $T' = (\Sigma,Q',I',\delta')$ as follows: $Q' = Q \times \Pi \times [0,k-1]$, $I' = I \times \{\pi_0\} \times \{0\}$ where $\pi_0 \in \Pi$ is some arbitrary fixed designated permutation\footnote{For example, if $\Gamma$ has a linear ordering then a convenient choice for $\pi_0$ may be the permutation that orders the labels from smallest to largest. In any case, any other permutation will do just as well.},
and $\delta'$ is defined by: 
\[
\delta'((q, \pi, h), a) =  (\delta(q,a), \pi \Lsh \lambda(q), maxPos(\lambda(q), \pi))
\]
 
Given a state $s = (q, \pi, h)$ of  $\AUTB$, we call $q$ the \emph{state in} s, call $\pi$ the \emph{permutation in} $s$, and call $h$ the \emph{pointer in} $s$. Given in addition a set $X \subseteq \Gamma$, we will say that $s$ \emph{displays} $X$ if $\pi$ displays $X$, and that $s$ \emph{separates} $X$ if $(\pi, h)$ separates $X$.
Finally, the coloring function $c$ of $\AUTB$ assigns to a state $s$ of $\AUTB$ one of the colors $[1,2k]$ as follows: let $X \subseteq \Gamma$ be the set separated by $s$, and define $c(s) = 2|X| -1$ if $B(X) = \false$, and $c(s) = 2|X|$ if $B(X) = \true$.

Intuitively, the job of the permutations and pointers in the states of $\AUTB$ is to help the parity condition of $\AUTB$ detect if the set of labels visited infinitely often by the corresponding run of $\AUTA$ satisfies the Boolean function $B$ in the EL condition. Observe that it is because of the word `infinitely' in the previous sentence that the permutation (and the pointer) in the initial state can be chosen arbitrarily. Informally, the way the transition relation of $\AUTB$ works is as follows: reading a symbol $a$ at a state $s = (q, \pi, h)$, the first coordinate (i.e., the state of $\AUTA$) is updated according to the transition relation of $\AUTA$, the permutation is updated by pushing all the labels that the source state $q$ visits to the left, and the pointer is set to the right-most position in $\pi$ that contains a label visited by~$q$.

Note that, given a run $\varrho$ of $\AUTB$ one can obtain a \emph{corresponding} run $\rho$ of $\AUTA$ by projecting each state in $\varrho$ onto its first coordinate.
Given a run $\varrho$ of $\AUTB$, and $i \in \mathbb{N}$, we will denote by $q_i, \pi_i, h_i$ the three components that make up the state $\varrho_i$.

To prove that the languages of the EL-automaton $\AUTA$ and the constructed parity automaton $\AUTB$ are the same we will need the following lemma.
Intuitively, this lemma states that (due to the fact that every time a label is visited it is pushed to the left in the next permutation) for every run of $\AUTA$, some time after the labels that appear on the run only finitely often are no longer visited, a point is reached after which the labels that are visited infinitely often by the run always appear to the left of the labels that are not visited anymore; Furthermore, infinitely many times the pointer points at the right most element of the labels that appear infinitely often.

\begin{lemma}\label{lem:display and separate}
  Given a run $\varrho$ of $\AUTB$, and the corresponding run $\rho$ of $\AUTA$, we have that:
  \begin{enumerate}
    \item There is an $i \in \mathbb{N}$ such that $\varrho_j$ displays $\inf_\lambda(\rho)$ and $h_j < |\inf_\lambda(\rho)|$, for every $j \geq i$.
    \item For every $n \in \mathbb{N}$ there is an $m > n$ such that $\varrho_m$ separates $\inf_\lambda(\rho)$.
  \end{enumerate}
\end{lemma}
\begin{proof}
    To see the first item, given $x \in \inf_\lambda(\rho)$ and $y \in \Gamma \setminus \inf_\lambda(\rho)$, let $t \in \mathbb{N}$ be the last time $y$ is visited on $\rho$, and let $l > t$ be some later time where $x$ is visited on $\rho$. By the definition of $\delta'$, the transition from $\varrho_l$ to $\varrho_{l+1}$ pushes $x$ to the left of every state not visited by $q_l$ and thus, in particular, to the left of $y$, i.e., $\pi_{l+1}^{-1}(y) > \pi_{l+1}^{-1}(x)$. Furthermore, for every $l' > l$ we have that $x$ remains to the left of $y$ in $\pi_{l'}$ (indeed, the only way for $y$ to possibly move to the left of $x$ is if $y$ is visited between time $l$ and $l'$, but by our choice of $l$ this cannot happen). Since such an $l$ exists for every pair $x \in \inf_\lambda(\rho)$ and $y \in \Gamma \setminus \inf_\lambda(\rho)$, and since there only finitely many such pairs, we conclude that by taking a large enough $i$ we get that for every $j \geq i$ we have ($\dagger$) (i) $\lambda(q_j) \subseteq \inf_\lambda(\rho)$ and (ii) $\pi_j$ has the labels in $\inf_\lambda(\rho)$ appear to the left of the labels in $\Gamma \setminus \inf_\lambda(\rho)$.
    It remains to show the second part of the first item. Assume by contradiction that for some $j > i$ we have that $h_j \geq |\inf_\lambda(\rho)|$. Recall that by our assumption $j-1 \geq i$ and thus, the value of $h_j$ and the second part of $\dagger$ imply that $\pi_{j-1}(h_j) \not \in \inf_\lambda(\rho)$. But this is a contradiction to the first part of $\dagger$ since, by the definition of $\delta'$, we know that $q_{j-1}$ visits $\pi_{j-1}(h_j)$.

    To see the second item of the lemma, let $i$ be as in the first item of the lemma, and observe that we can assume w.l.o.g. that $n > i$. Let $h = |\inf_\lambda(\rho)|-1$, and let $x = \pi_n(h)$. Note that, $n > i$ implies, by the first item of the lemma, that $x \in \inf_\lambda(\rho)$. Let $m > n$ be such that $q_m$ visits $x$. There are two cases: either $q_j$ visits $\pi_j(h)$ for some $n \leq j < m$ or not. In the first case we are done since this implies that $h_{j+1} \geq h$ and thus, by the first item of the lemma, it must be that $\varrho_{j+1}$ separates $\inf_\lambda(\rho)$.
    For the case that $q_j$ does not visit $\pi_j(h)$ for every $n \leq j < m$, observe that by the first item of the lemma we have that $q_j$ also does not visit any of the labels that $\pi_j$ maps to positions larger than $h$, since these are all outside $\inf_\lambda(\rho)$, and are thus not visited by states after $\rho_i$ (recall that we assume that $n > i$). Hence, by the definition of $\delta'$, it must be that $\pi_{j+1}(l) = \pi_j(l)$ for every $l \in [h, k-1]$, and every $n \leq j < m$ (i.e., the left shifts that transform $\pi_j$ to $\pi_{j+1}$ only involve the labels in positions $0$ to $h-1$). Recall that $\pi_n(h) = x$ and conclude that also $\pi_m(h) = x$. Since $q_m$ visits $x$, and by the first item of the lemma $h_m \leq h$, it must be that $h_m = h$. Combining this with the fact that (again by the first item) $\varrho_m$ displays $\inf_\lambda(\rho)$, we get that $\varrho_m$ separates $\inf_\lambda(\rho)$, as promised.
\end{proof}

Lemma~\ref{lem:display and separate} directly yields the following corollary:
\begin{corollary}\label{cor:infinity set separated}
  If $\varrho$ is a run of $\AUTB$, and $\rho$ is the corresponding run of $\AUTA$, then: (i) $\inf_\lambda(\rho)$ is separated by infinitely many states in $\varrho$; (ii) every set $Y$ separated by infinitely many states in $\varrho$  satisfies $Y \subseteq \inf_\lambda(\rho)$.
\end{corollary}

We are now ready to prove Theorem~\ref{thm:EL to parity}.
\begin{proof}[Proof of Theorem~\ref{thm:EL to parity}]
  It is easy to see that the number of states and colors of the parity automaton $\AUTB$ is as stated. To see that it accepts the same language as $\AUTA$, take a run $\varrho$ of $\AUTB$ and let $\rho$ be the corresponding run of $\AUTA$. Let $z$ be such that $d \in [2z-1, 2z]$ is the highest color that appears infinitely often on $\varrho$. By the definition of the coloring function of $\AUTB$, it follows that (i) there is a set $X \subseteq \Gamma$ of size $z$ that is separated infinitely often by states on $\varrho$, and there is no set of larger size that is separated infinitely often by states on $\varrho$; (ii) $B(X) = \true$ if $d = 2z$, and $B(X) = \false$ if $d = 2z-1$. Note that, by Corollary~\ref{cor:infinity set separated}, it must be that $X = \inf_\lambda(\rho)$. Hence, $\rho$ is an accepting run iff $d$ is even, i.e., iff $\varrho$ is an accepting run.
\end{proof}

We remark that while we gave a construction for deterministic automata, the same construction with the trivial natural modifications will work for other types of transition systems (nondeterministic, alternating, etc.), preserving the type of transition system while converting the acceptance condition.


\section{Supplement to ``Reasoning''}


We begin with a remark the gives some standard relationships between satisfiability, validity, and model checking.

\begin{remark} \label{rem: relationships} 
A formula is valid iff its negation is not satisfiable. Thus, since the complexity classes in Theorems~\ref{thm: synthesis complexity} and~\ref{thm: satisfiability et al} are deterministic, and since \LTLfplus and \PPLTLplus are each closed under negation, validity has the same complexity as satisfiability. Also,  model checking is at least as hard as validity, since $\Psi$ is valid iff $N \models \Psi$ where $N$ is a "complete" transition system that generates all possible traces.
\end{remark}

We now discuss the lower bounds in Theorems~\ref{thm: satisfiability et al}.

\begin{enumerate}
    \item Satisfiability is EXPSPACE-hard already for the fragment of \LTLfplus consisting of formulas of the form $\safe{\Phi}$. The proof of this follows a standard encoding, cf.~\cite{DBLP:conf/stoc/VardiS85,kupferman2000automata,DBLP:conf/atva/BansalLTVW23}, see Theorem~\ref{thm: safe LTLf satisfiability is expspace-hard} below.

    \item  Model checking is EXPSPACE-hard already for the fragment of \LTLfplus consisting of formulas of the form $\guar{\Phi}$. This follows from Remark~\ref{rem: relationships} and the fact that validity for this fragment is EXPSPACE-hard; this can also be proven directly~\cite{DBLP:conf/atva/BansalLTVW23}.
    
    \item Satisfiability is PSPACE-hard already for the fragment \PPLTLplus consisting of formulas of the form $\guar{\Phi}$. 
    We reduce \LTLf satisfiability, a known PSPACE-hard problem~\cite{DegVa13}, to \PPLTL satisfiability, and then to \PPLTLplus satisfiability. For an \LTLf formula $\Phi$, let $\overleftarrow{\Phi}$ be the \PPLTL formula formed by syntactically replacing every future operator by its past variant, i.e., replacing $\LTLX$ by $\LTLY$, and $\LTLU$ by $\LTLS$. Then, $\Phi$ and $\overleftarrow{\Phi}$ are equisatisfiable: a finite trace $u = u_0 \cdots u_n$ is a model of $\Phi$ iff its reverse $u_n u_{n-1} \cdots u_0$ is a model of $\overleftarrow{\Phi}$. Thus, also $\overleftarrow{\Phi}$ and $\guar{\overleftarrow{\Phi}}$ are equisatisfiable (extend a model of $\overleftarrow{\Phi}$ arbitrarily to get a model of $\guar{\overleftarrow{\Phi}}$, the other direction is by definition of $\guarSymbol$). 

    \item  Model checking is PSPACE-hard already for the fragment of \PPLTLplus consisting of formulas of the form $\safe{\Phi}$. This follows from Remark~\ref{rem: relationships} (that model-checking is at least as hard as validity), and the fact that the previous item implies that validity of \PPLTLplus formulas of the form $\safe{\Phi}$ is PSPACE-hard. 
\end{enumerate}

\subsection{Solution techniques for reasoning}

We continue where the body left off and show how to perform $(\dagger)$ (a) and (b). First we show how to do this for \LTLfplus, and then for \PPLTLplus.

To handle an \LTLfplus formula $\Psi$ (assumed to be in positive normal-form), we convert each finite-trace formula $\Phi_i$ into an equivalent NFA $\AUTB_i$ of exponential size using Theorem~\ref{thm:logic to FA}. We do not convert the NFA into an equivalent DFA $\AUTA_i$ as in Step 1 of the Synthesis Algorithm, since it may be of double-exponential size. Instead, we run the ``on-the-fly" algorithm described in the body. Note that a state $s$ of $D$ is of the form $(q_1,\cdots,q_n)$ where each $q_i$ is a state in the DFA $\AUTA_i$, which itself is a subset of the set of states of the NFA $\AUTB_i$. Thus, storing a single state $s$ of $D$ can be done with linearly many states of the NFAs, which uses exponential space in the size of $\Psi$. Also, storing the labels of a state can be done with $k = O(|\Psi|)$ bits. To see that (a) of $(\dagger)$ is satisfied we just need to identify if a state $q_i$ in $\AUTA_i$ is a final state. This can be done efficiently since $q_i$ is a final state of $\AUTA_i$ iff it contains a final state of the NFA $\AUTB_i$. To see that (b) of $(\dagger)$ is satisfied observe that to decide if there is an $a$-labeled transition in $D$ from $s = (q_1,\cdots,q_n)$ to $s' = (q'_1,\cdots,q'_n)$ it is enough to check that, for each $i$, there is an $a$-labeled transition in $D_i$ from $q_i$ to $q'_i$. The latter can be efficiently checked by consulting the NFA's transition relation as in Remark~\ref{rem: LTLf to DFA}: for every NFA state $s'$, check that $s' \in q'_i$ iff there is an NFA transition $(s,a,s')$ for some $s \in q_i$.

To handle a \PPLTLplus formula $\Psi$ (assumed to be in positive normal-form), we do not convert each $\Phi_i$ into an equivalent DFA $\AUTA_i$ as in Step 1 of the Synthesis Algorithm. Instead, we run the ``on-the-fly'' algorithm described above. Note that a state $s$ of $D$ is of the form $(q_1,\cdots,q_n)$ where each $q_i$ is a state of the DFA $\AUTA_i$, and by Remark~\ref{rem: PP to DFA} each state $q_i$ is a set of subformulas of $\Phi_i$. Thus, each state of $D$ can be stored with a polynomial number of bits. Also, storing the labels of a state can be stored with $k = O(|\Psi|)$ bits. Closely following the case of \LTLfplus above, to show $(\dagger)$ we just need to efficiently (1) 
identify if a state $q_i$ in $\AUTA_i$ is final, and (2) check if there is an $a$-labeled transition from $q_i$ to $q'_i$ in $\AUTA_i$. This is true for (1) since $q_i$ is a final state iff it contains $\Phi_i$; and is true for (2) since we can easily check, by consulting $\Phi_i$, that a subformula $s'$ of $\Phi_i$ is in $q'_i$ iff $q_i$ yields $s'$ under $a$ in one step, by recursive evaluation, as in Remark~\ref{rem: PP to DFA}.

This establishes that satisfiability for \LTLfplus (resp. \PPLTLplus) is in NEXPSPACE (resp. NPSPACE). Now use the fact that PSPACE=NPSPACE and EXPSPACE=NEXPSPACE~\cite{Savitch70} to get the upper bounds in Theorem~\ref{thm: satisfiability et al}.



\subsubsection{Model Checking} 
We outline the algorithms for model checking. 
Given a nondeterministic transition system $T$ and a \LTLfplus/\PPLTLplus formula $\Psi$, we have that $T$ satisfies $\Psi$ iff it is not the case that there is an infinite trace generated by $T$ that satisfies $\lnot \Psi$. Thus, since deterministic complexity classes are closed under complement, it is enough to provide an optimal algorithm for the following problem: given a nondeterministic transition system $T$ and a \LTLfplus/\PPLTLplus formula $\Psi$, decide if there is an infinite trace generated by $T$ that satisfies $\Psi$. 
To do this we proceed as in the ``on-the-fly'' algorithm for  satisfiability. However, we add to the current state of $D$ a component consisting of the current state of $T$. Thus, we store $(q_1,\cdots,q_n,t)$. To check if there is a transition from 
$(q_1,\cdots,q_n,t)$ to $(q'_1,\cdots,q'_n,t')$ we guess an input letter $a$, check there is a transition from $(q_1,\cdots,q_n)$ to $(q'_1,\cdots,q'_n)$ on $a$ and a transition from $t$ to $t'$ on $a$. Since $T$ is part of the input, this only uses an additional space that is linear in the size of $T$. Thus, this algorithm decides if there is an infinite trace generated by $T$ whose induced run in $D$ satisfies the EL-condition, i.e., by Proposition~\ref{prop: product}, it decides if there is an infinite trace generated by $T$ that satisfies $\Psi$.

\subsection{Proof that Satisfiability for \LTLfplus  is EXPSPACE-hard} 

Because it does not seem to explicitly appear in the literature, for completeness we supply a proof of the following:

\begin{theorem} \label{thm: safe LTLf satisfiability is expspace-hard}
The satisfiability problem for the fragment of \LTLfplus consisting of formulas of the form $\safe{\Phi}$ is EXPSPACE-hard.
\end{theorem}

The proof uses an encoding from~\cite{DBLP:conf/stoc/VardiS85}, also found in \cite{kupferman2000automata,DBLP:conf/atva/BansalLTVW23}.
Recall that EXPSPACE is exactly the class of languages accepted by deterministic Turing Machines with space bound $2^{cn}$ for non-zero constant $c$.  Let $T$ be such a Turing Machine (TM), say with tape alphabet $\Gamma$, and state set $Q$. Let $x$ be an input word, say of length $n$. For technical convenience, we assume $c \geq 2$ to ensure that even if $n=1$ each configuration consists of at least $4$ tape cells. We will construct an \LTLf formula $\Phi$ of length polynomial in $n$ such that the TM accepts $x$ iff $\safe{\Phi}$ is satisfiable.\footnote{Recall that for such a reduction we consider the size of the TM as a constant, and only consider the size of the formula with respect to the size of the input word $x$.} As usual, we assume that the TM halts on every input, in either the halt-reject state or the halt-accept state. We interpret a halting configuration as if it has a transition that does nothing.

The \LTLf formula $\Phi$ we build will be such that an infinite trace $\tau$ will satisfy $\safe{\Phi}$ iff $\tau =  \tau' v^\omega$ where $\tau'$ is the encoding of an accepting run of the TM on $x$, and $v$ is the encoding of the last configuration of that run. We now describe the encoding. Define a new alphabet $\Sigma = \Gamma \cup (\Gamma \times Q) \cup \{0,1\} \cup \{\#,\$\}$. We encode a cell (of the TMs tape) by a word of length $cn+2$ in $\{0,1\}^{cn} (\Gamma \cup (\Gamma \times Q)) \{\#,\$\}$: the first $cn$ characters encode in binary (least significant bit first) the position, aka \emph{index}, of the cell within the tape (thus, in a configuration, the first index represents the number $0$, and the last index represents the number $2^{cn}-1$); the next character encodes the \emph{contents} of the cell, including whether or not the TM head is there, and if so, also the current state; and the final character is the \emph{cell-delimiter} $\#$ if the index of the cell is less than $2^{cn}-1$, and is the \emph{configuration-delimiter} $\$$ otherwise. A configuration is encoded by a concatenation of $2^{cn}$ cell encodings.
A run, or a partial run, is encoded by the configuration-delimeter $\$$ followed by a concatenation of configuration-encodings.




We now describe $\Phi$, mainly in words --- translating this description into an \LTLf formula is straightforward and standard, e.g.,~\cite{kupferman2000automata,DBLP:conf/atva/BansalLTVW23}. For $m \geq 0$, let $\LTLX^m$ be a shorthand for $\LTLX \LTLX \cdots \LTLX$ where the number of operators is $m$. To express, of a finite trace, that there
are at least $m$ more symbols before the end, one can write $\LTLX^m \true$. To express of a finite trace that there are exactly $m$ more symbols before the end, one can write $\LTLX^m \last$ where $\last$ is shorthand for $\lnot \LTLX \true$.

The formula $\Phi$ will be the conjunction of the following:
\begin{enumerate}
\item ``The first symbol is $\$$";

\item ``Always, if the current symbol is a delimiter, and there are at least $cn+2$ more symbols, then the next $cn+2$ symbols are of the from $\{0,1\}^{cn} (\Gamma \cup (\Gamma \times Q)) \{\#, \$\}$'';

\item ``Always, if the current symbol is a delimiter, and  there is another delimiter in the future, then the current cell ends with $\$$ if its index is $2^{cn}-1$, and with $\#$ otherwise'';

\item ``Always, if the current symbol is $\$$, and there is another delimiter later on, then the index of the current cell is $0$'';

\item ``Always, if the current symbol is a delimiter, and if there are two delimiters in the future, then the index of the current cell plus $1$ modulo $2^{cn}$ is the index of the next cell''; 

\item[] A naive formula would disjunct over all $cn$ bitstrings and check that the next index is a $cn$ bitstring representing the successor. However, there are $2^{cn}$ such bitstrings, so the size of such a formula would be exponential in $n$. Instead, one can check for an increment with a polynomial sized formula that expresses the following: consider the largest block of $1$s in the current index (if any), flip all these to $0$ in the next index, and the immediately following $0$ in the current index (if any) flips to $1$ in the next index. 

\item ``If there is at least one $\$$ after the initial $\$$, then the content of the first cell contains the head at the initial state, and after that there is no cell with the head until a $\$$, and the content of the first $cn$ cells contain $x$, and the contents of all the cells after that until a $\$$ encode empty cells'';

\item ``Always, if the current symbol is a delimiter, and  the word ends with a delimiter, and there is a $\$$ between these two delimiters, and the index of the third cell from the end of the word is equal to the index of the current cell, then the current triple of cell contents constrains the last triple of cell contents according to the TM''; 

\item[] First, we remark that a triple $t$ of adjacent cells constrains the triple $t'$ of cells with in the same position in the next configuration as follows: if the head appears $t$ then the triple $t'$ is determined according to the TM's transition function; if the head does not appear in $t$ then every cell $t$ should have its content copied to the corresponding cell in $t'$ (with the same index) if one can deduce that it is not adjacent to the head (i.e., if it is the middle cell, or it is the first two cells and the triple starts at position $0$, or it is the last two cells and the triple ends at position $2^{cn}-1$). Note that even though there are exponentially many (in $n$) possible triples of cells, the content of a triple of cells has constant size, it consists of three symbols from $\Sigma$.

\item[] Second, we remark that one can compare two indices for equality with a polynomially-sized formula provided the second index starts a constant number $m \geq cn$ of positions before the end of the string. This is done by the formula:
\[
\bigwedge_{i \in [0,cn-1]} (\LTLX^i 1) \leftrightarrow \eventually (1 \land \LTLX^{m-i} \last).
\] 


\item ``Always, the current symbol does not include the halt-reject state''.
\end{enumerate}


To see that this is correct, we argue as follows. 
Let $\tau'$ be an encoding of the run of the TM on $x$, and let $v$ be its last configuration. If the TM accepts $x$ then, it is not hard to see that $\tau = \tau'v^\omega$ satisfies $\safe{\Psi}$. On the other hand, if the TM rejects $x$, then if $\tau = \tau'v^\omega$, then for a long enough prefix of $\tau$, we have that $\safe{\Psi}$ would not be satisfied by item 8; and if $\tau$ is not equal to $\tau'v^\omega$, i.e., it doesn't faithfully encode the run of the TM on $x$ (extended by copying the last configuration forever), then this violation of the encoding will be detected on a long enough prefix by one of the other seven items, and thus $\tau$ would not satisfy $\safe{\Psi}$.

\end{document}